\documentclass[journal]{IEEEtran}

\usepackage{cite}

\usepackage{graphicx}
\usepackage{psfrag}

\usepackage[cmex10]{amsmath}
\usepackage{amssymb}
\usepackage{amsthm}

\def\tr{\mathrm{tr}}

\def\diag{\mathrm{diag}}

\DeclareMathOperator{\prob}{\mathbb{P}}

\newcommand{\vect}[1]{\mathbf{#1}}
\newcommand{\bx}{\mathbf{x}}
\newcommand{\br}{\mathbf{r}}
\newcommand{\bz}{\mathbf{z}}

\newtheorem{theorem}{Theorem}
\newtheorem{lemma}{Lemma}
\newtheorem{corollary}{Corollary}
\newtheorem{example}{Example}
\newtheorem{remark}{Remark}

\begin{document}
%
\title{Large System Analysis of the Energy Consumption Distribution in Multi-User MIMO Systems with Mobility}
%
%
%

\author{Luca~Sanguinetti,~\IEEEmembership{Member,~IEEE,}
        Aris~L.~Moustakas,~\IEEEmembership{Senior Memebr,~IEEE,}\\
        Emil~Bj{\"o}rnson,~\IEEEmembership{Member,~IEEE,}
        and~Merouane Debbah,~\IEEEmembership{Fellow,~IEEE}
\thanks{\newline \indent L. Sanguinetti is with the
University of Pisa, Dipartimento di Ingegneria dell'Informazione, Pisa, Italy (luca.sanguinetti@iet.unipi.it) and is also with Sup\'elec, Gif-sur-Yvette, France. A. L. Moustakas is also with Department of Physics, National \& Capodistrian University of Athens, Athens, Greece (arislm@phys.uoa.gr). E. Bj\"ornson with ISY, Link\"{o}ping University, Link\"{o}ping, Sweden (emil.bjornson@liu.se).  M. Debbah is with Sup\'elec, Gif-sur-Yvette, France (merouane.debbah@supelec.fr). 
\newline\indent This research has been supported by the FP7 Network of Excellence in Wireless COMmunications NEWCOM\# (Grant agreement no. 318306). L.~Sanguinetti is also funded by the People Programme (Marie Curie Actions) FP7 PIEF-GA-2012-330731 Dense4Green. A. L. Moustakas is the holder of the DIGITEO "ASAPGONE" Chair. The research of M. Debbah has been also supported by the ERC Starting Grant 305123 MORE. 
\newline \indent Parts of this paper were presented at the IEEE International Conference on Acoustics, Speech and Signal Processing (ICASSP), Florence, Italy, May 4 -- 9, 2014 \cite{Sanguinetti2014_ICASSP}.
}}

\maketitle

\begin{abstract}
In this work, we consider the downlink of a single-cell multi-user MIMO system in which the base station (BS) makes use of $N$ antennas to communicate with $K$ single-antenna user equipments (UEs). The UEs move around in the cell according to a random walk mobility model. We aim at determining the energy consumption distribution when different linear precoding techniques are used at the BS to guarantee target rates within a finite time interval $T$. The analysis is conducted in the asymptotic regime where $N$ and $K$ grow large with fixed ratio under the assumption of perfect channel state information (CSI). Both recent and standard results from large system analysis are used to provide concise formulae for the asymptotic transmit powers and beamforming vectors for all considered schemes. These results are eventually used to provide a deterministic approximation of the energy consumption and to study its fluctuations around this value in the form of a central limit theorem. Closed-form expressions for the asymptotic means and variances are given. Numerical results are used to validate the accuracy of the theoretical analysis and to make comparisons. We show how the results can be used to approximate the probability that a battery-powered BS runs out of energy and also to design the cell radius for minimizing the energy consumption per unit area. The imperfect CSI case is also briefly considered.
\end{abstract}


%
\IEEEpeerreviewmaketitle

\section{Introduction}
\IEEEPARstart{T}{he data}  traffic in cellular networks has been increasing exponentially for a long time and is expected to continue this trend, at least for the next five years \cite{Cisco11}. Currently, one of the biggest challenges related to the traffic growth is the increasing energy consumption of the cellular infrastructure equipments \cite{Fehske11}. This means that the energy consumption must be taken into account from the outset when designing cellular networks of the future. This is particularly important when deploying base stations (BSs) in new rural regions of the world, where the electrical grid is unreliable or even non-existing. Off-grid deployments rely on combinations of diesel generators, batteries, and local energy harvesting  (e.g., from solar panels) \cite{Fehske11}. Since the supply of energy is either limited or fluctuates with the harvesting, it is of paramount importance to operate the BS such that it will not run out of energy (also known as power outage).

A lot of work has been done recently to better understand the energy consumption tradeoffs in cellular networks \cite{Chen2011}. However, most of the analysis is carried out without taking into account the impact of user equipment (UE) mobility. This is a major issue since UE movements play a key role in determining the required energy in outdoor scenarios. Consider, for example, the case in which all UEs are positioned at the cell edge and assume that some target rate must be guaranteed at each UE. This would result into an increased power consumption, which might rapidly make the BS run out of energy if the UEs remain at the cell edge for an extended period of time (due to the lack of mobility). Therefore, it is of major importance to quantify not only how often such unlikely events happen, but also how long they last. All this depends on the UE mobility model and in particular on how their positions, velocities and accelerations change over time. A first attempt in this direction can be found in \cite{Decreusefond13} in which the authors make use of stochastic geometry to model the energy consumption of a cellular network where each UE is connected to its closest single-antenna BS. In particular, a refined energy consumption model is proposed that includes the energy of broadcast messages, traffic activity, and user mobility. 

The theory of random walks has been analyzed extensively in probability theory and it has been applied to a variety of fields. In a nutshell, it is a mathematical formalization of a path that consists of a succession of random steps \cite{paul_book}. If the time elapsed between two successive steps is relatively small (possibly zero), then the central limit theorem implies that the distribution of the corresponding path-valued process approaches that of a Brownian motion \cite{morters_book}. A more involved model (known as L\'evy flight) corresponds to a random walk in which the step lengths have a probability distribution that is heavy tailed \cite{paul_book}. Another one that has been widely investigated is the so-called random waypoint process in which a mobile user picks a random destination and travels with constant velocity to reach it. {Although human mobility is rarely random, in \cite{Rhee2011_LevyHumanMobility, Scafetta2011_LevyWalkHumanMobility} the authors show that a L\'evy flight contains some important statistical similarities with human mobility. This suggests that a random walk model might be considered as a good one to achieve a reasonable tradeoff between accuracy and analytical tractability \cite{Rhee2011_LevyHumanMobility, Scafetta2011_LevyWalkHumanMobility}. Despite the well-understood properties of the random mobility models above, not much progress
has been made towards providing analytical results for the statistics of metrics in wireless communications.}

Motivated by this lack, in this work we consider the downlink of a single-cell multi-user multiple-input multiple-output (MIMO) system. The BS is equipped with an array of $N$ antennas and serves simultaneously $K$ UEs by using linear precoding. The UEs are assumed to move around in the cell following (independent) random walks. Based on the random walk mobility pattern, we aim at determining the statistical distribution of the energy consumption required to guarantee a given set of UE data rates within a finite time interval of length $T$. {The analysis is conducted in the asymptotic regime where $N$ and $K$ grow large with fixed ratio $c = K/N$ and it is applied to several linear precoding schemes, under the assumption of perfect channel state information (CSI) at the BS.} In particular, we first focus on the asymptotic design of the precoder minimizing the transmit power at each time slot while satisfying the rate requirements \cite{Wiesel06,Bjornson2012a}. Differently from \cite{Huang2012,Zakhour2012,Asgharimoghaddam14}, the asymptotic transmit powers and beamforming vectors are computed using the approach adopted in \cite{Romain2014}, which provides us a much simpler means to overcome the technical difficulties arising with the application of standard random matrix theory tools (e.g. \cite{Zakhour2012}). As already pointed out in \cite{Huang2012,Zakhour2012,Asgharimoghaddam14}, in the asymptotic regime the optimal values can be computed in closed-form through a nice and simple expression, which depends only on the average channel attenuations and rate requirements. The asymptotic analysis is then applied to design different heuristic linear precoding techniques, namely, maximum ratio transmission (MRT), zero-forcing (ZF), and regularized ZF (RZF) \cite{Wagner12,Evans2013,Sanguinetti2014_Globecom,Bjornson:2013kc}. {In particular, we concentrate on a RZF scheme (e.g. \cite{Evans2013,Sanguinetti2014_Globecom}) that exploits knowledge of the average channel attenuations of UEs in designing the beamforming vectors and provide the closed-form expression for the optimal regularization parameter in the asymptotic regime.} 

The asymptotic analysis above is then used to provide a deterministic approximation of the energy consumption (which is asymptotically accurate in the large system limit)  and to study its fluctuations around this value under the form of a central limit theorem. We compute closed-form expressions for the asymptotic means and variances of all considered schemes exploiting the connection between random walk and Brownian motion. It turns out that the fluctuations induced by the randomness in the UE positions make the variance of the energy decrease as $1/K$ for all considered schemes. This makes the variability induced by the small-scale fading less important for the analysis since it scales as $1/K^2$ \cite{Sanguinetti2014_ICASSP,Bai2004_CLT_covariance_matrices}. The analytical expressions are validated by means of extensive simulations and are shown to closely match the numerical results for different settings. Finally, we exemplify how the new statistical characterization can be used to compute the probability that a battery-powered BS runs out of energy within a certain time period and also to design the cell radius for minimizing the energy consumption per unit area.

The remainder of this paper is organized as follows. \footnote{The following notation is used throughout the paper. The notations ${\mathsf{E}}_{\mathbf{z}}\{\cdot\}$, {${\mathsf{VAR}}_{\mathbf{z}}\{\cdot\}$ and ${\mathsf{COV}}_{\mathbf{z}}\{\cdot\}$} indicate that the expectation, variance and covariance are computed with respect to $\mathbf{z}$. The notation $||\cdot||$ stand for the Euclidean norm whereas $|\mathcal S|$ is used to denote the cardinality of the enclosed set $\mathcal S$. We let $\mathbf{I}_K$ denote the $K \times K$ identity matrix, whereas ${\bf{1}}_K$ and ${\bf{0}}_K$ are the $K$-dimensional unit and null column vectors, respectively. We use $\mathcal{CN}(\cdot,\cdot)$ to denote a multi-variate circularly-symmetric complex Gaussian distribution whereas $\mathcal{N}(\cdot,\cdot)$ stands for a real one. The notation $\mathop  {\longrightarrow} \limits^{\mathcal D}$ denotes convergence in distribution whereas $\mathop  {\longrightarrow} \limits^{a.s.}$ stands for almost surely equivalent. We denote $J_n(\cdot)$ the $n$-order Bessel function and use $Q(z)$ to indicate the Gaussian tail function while $Q^{-1}(z)$ indicates its inverse. We use $\nabla f (\bz)$ to indicate the vector differential operator whereas $\nabla^2 f(\bz)$ stands for the Laplace operator. The notation $x \sim y$ with $x,y$ real-valued numbers means that $x$ is approximately equal to $y$.} The next section introduces the system model and formulates the problem under investigation. Section \ref{Optimal_precoding} revisits the optimal linear precoder design in both finite and asymptotic regimes. Section \ref{Heuristic_precoding} deals with the asymptotic analysis and design of MRT, ZF and RZF. All this is then used in Section \ref{Energy_fluctuations} to prove that in the asymptotic regime the energy consumption for all considered schemes converges to a Gaussian random variable whose mean and variance are derived analytically. In Section \ref{Numerics}, numerical results are used to validate the theoretical analysis and to make comparisons among different processing schemes. Some applications of the statistical characterization of the energy consumption are also exemplified in Section \ref{Numerics}. {In addition, the imperfect CSI case is briefly addressed for the ZF and RZF precoding schemes.} Finally, some conclusions and discussions are drawn in Section \ref{Conclusions}.

\section{System model and problem statement}
\subsection{System model}
We consider the downlink of a single-cell multi-user MIMO system in which the BS makes use of $N$ antennas to communicate with $K$ single-antenna UEs over the whole channel bandwidth of $W$ Hz. Linear processing is used for data precoding and {perfect CSI is assumed to be available at the BS. As mentioned above, the imperfect CSI case will be briefly addressed in Section \ref{imperfect_CSI}.}  The $K$ active UEs are moving around within the coverage cell $\mathcal C$ of area $|\mathcal C|$. We assume that the number of UEs $K$ can increase arbitrarily, while the area $|\mathcal C|$ is maintained fixed. 
To simplify the computations, we assume that the cell is circular with radius $R$ such that $|\mathcal C|= \pi R^2$ and that the BS is located in the centre of the cell with its $N$ transmit antennas adequately spaced apart (such that the channel components to any UE are uncorrelated). The location of UE $k$ at time $t$ is denoted by $\mathbf{x}_k(t) \in \mathbb{R}^2$ (in meters) and it is computed with respect to the BS. We call ${\bf{X}}_k = \{\mathbf{x}_k(t); t\ge 0\}$ and model ${\bf{X}}_k$ and ${\bf{X}}_j$ for $k\ne j$ as obtained from identical and independent random walks constrained into the region $\mathcal C$ \cite{Camp02asurvey}. We let $\ell$ be the (average) size of a step of the random walk and let $\xi$ be the corresponding (average) time elapsed between two successive steps. The (average) velocity of each user between consecutive steps is thus given by $\sim \ell/\xi$. We assume that after each step the user performs a next step in an independent random direction.\footnote{The model can be generalized to include correlated steps \cite{morters_book}. In this case, $\ell$ and $\xi$ will correspond to the length and time over which the correlation in the direction of the walker is lost.} Thus, ${\bf x}_k(t)$ can be expressed as a sum over $\left\lfloor {t/\xi} \right\rfloor$ independent and identically distributed zero-mean steps $\Delta {\bf x}_k(i)$, i.e., $  {\bf x}_k(t) = \sum_{i=1}^{\left\lfloor {t/\xi} \right\rfloor} \Delta{\bf x}_k(i)$
with $|\Delta{\bf x}_k(i)| = \ell$. As it is known (see Appendix B for a short discussion on this), if $\ell \to 0$ and $\xi \to 0$ with $\ell^2/\xi$ kept fixed, then the random walk converges (due to the central limit theorem) to a Brownian motion with diffusion coefficient $D$ such that $4D = \ell^2/\xi$.

We call $\mathbf{h}_{k}(t)= [{h}_{k,1}(t),{h}_{k,2}(t),\ldots,{h}_{k,N}(t)]^T \in \mathbb{C}^{N\times1}$ the vector whose generic entry $h_{k,n}(t)$ represents the channel propagation coefficient between the $n$th antenna at the BS and the $k$th UE at time $t$. In particular, we assume that
\begin{align}\label{h_k}
\mathbf{h}_{k}(t)= \sqrt{l(\mathbf{x}_{k}(t))}\mathbf{w}_{k}(t)
\end{align}
where $\mathbf{w}_{k}(t) \sim \mathcal {CN} (0,{\bf I}_K)$ accounts for the small-scale fading channel and {$l(\cdot): \mathbb{R}^2 \to \mathbb{R}^+$ describes the large-scale channel fading at different user locations; that is, $l(\mathbf{x}_{k}(t))$ is the average channel attenuation due to pathloss and shadowing at distance $||\mathbf{x}_{k}(t)||$}. {As seen, $l(\mathbf{x}_{k}(t))$ is assumed to be independent from the transmit antenna index $n$. This is a reasonable assumption since the distances between UEs and BS are much larger than the distance between the antennas \cite{Calcev2007}.} Since the forthcoming analysis does not depend on a particular choice of $l(\cdot)$ as long as it is a function of the user distance and is bounded from below, we keep it generic. {Moreover, in all subsequent discussions we assume that $\{l(\mathbf{x}_{k}(t)); k=1,2,\ldots,K\}$ are known at the BS (see Section \ref{imperfect_CSI} for a brief discussion on this).}

\begin{example}\label{example_1}
{In the simulations, we assume that the average channel attenuation at a generic position $\mathbf{x}$ is dominated by the pathloss and is evaluated as \cite{Adhikary2013}}
\begin{align}\label{avg_pathloss_1}
l(\mathbf{x}) =  2L_{\bar x}\left({1+ \frac{\|\mathbf{x}\|^{\beta}}{\bar x^{\beta}}}\right)^{-1}
\end{align}
where $\beta\ge 2$ is the pathloss exponent, $\bar x > 0$ is some cut-off parameter and $L_{\bar x}$ is a constant that regulates the attenuation at distance $\bar x$. The average inverse channel attenuation $\mathbb{E}_\mathbf{x} \{  l^{-1}(\mathbf{x})\}$ plays a key role in all subsequent discussions. Simple integration shows that, if $l(\mathbf{x})$ is given by \eqref{avg_pathloss_1} and the initial positions of UEs are assumed to be uniformly distributed within the cell, then
\begin{align}\label{avg_pathloss}
\frac{1}{|\mathcal C|}\int_{\mathcal C} \frac{1}{l(\mathbf{x})}d{\bf x} = \frac{R^{\beta}}{2\bar x^{\beta} L_{\bar x}}\left(\frac{2}{2+\beta} + \frac{\bar x^{\beta}}{R^{\beta}}\right).
\end{align}
\end{example}

We assume that the temporal correlations of $\{\mathbf{w}_{k}(t); t \ge 0\}$ are such that the coherence time is $\Delta\tau\sim \lambda\xi/\ell$, where $\lambda$ is the wavelength and $\ell/\xi$ is the velocity of each UE due to the random walk. 

\begin{example}\label{example_2}
Consider a cell with radius $R= 500$ m and operating at the carrier frequency of $f_c = 2.4$ GHz so that $\lambda = 0.125$ m. Assume $\ell = 50$ m and $\xi = 30$ seconds. Therefore, the average velocity of each UE within two consecutive time steps is $\ell/\xi = 100$ m/minute and the diffusion coefficient is $D = \ell^2/(4\xi) = 1250$ m$^2$/minute. Under these assumptions, the coherence time $\Delta\tau$ of the small-scale fading is approximately $\Delta\tau \sim \lambda \xi/\ell = 0.075$ s. On the other hand, the coherence time of UE movements (also known as forgetting time) is $\sim R^2/D=3.33$ hrs, which is much larger than $\Delta\tau$. All this will be used later on to give some intuitions about the different impact of small-scale fading and UE mobility on the energy fluctuations.
\end{example}

\subsection{Problem statement}
We call $\mathbf{s}(t)  = [s_1(t),s_2(t),\ldots,s_K(t)]^T\in \mathbb{C}^{K\times 1}$ the signal transmitted at time $t$ and denote by $\mathbf{V}(t) = [\mathbf{v}_1(t),\mathbf{v}_2(t),\ldots,\mathbf{v}_K(t)]\in \mathbb{C}^{N\times K}$ its precoding matrix. We assume that $\mathbf{s}(t)$ {originates from a Gaussian codebook with zero mean and covariance matrix} $\mathbb{E}_{\mathbf{s}}[\mathbf{s}(t)\mathbf{s}^H(t)] = \mathbf{I}_K$. The sample ${y}_{k}(t)$ received at the $k$th UE at time $t$ takes the form
\begin{align}\label{y_k}
{y}_{k}(t)= \mathbf{h}_{k}^H(t)\mathbf{V}(t)\mathbf{s}(t)+{n}_{k}(t)
\end{align}
where ${n}_{k}(t) \sim \mathcal{CN}(0, \sigma^2)$ is the additive noise. Under the assumption of perfect CSI at UEs, the SINR at the $k$th UE is easily written as \cite{Bjornson:2013kc}
\begin{equation}\label{SINR_k}
{\rm{SINR}}_k(t) = \frac{\left|\mathbf{h}_k^H(t)\mathbf{v}_k(t)\right|^2}{\sum\limits_{i=1,i\ne k}^K\left|\mathbf{h}_k^H(t)\mathbf{v}_i(t)\right|^2 + {\sigma^2}}
\end{equation}
whereas the achievable rate in bit/s/Hz is given by 
\begin{equation}\label{r_k}
R_k(t) = \log_2 \left(1 + {\rm{SINR}}_k(t)\right).
\end{equation}
While conventional systems have large disparity between peak and average rates, we aim at designing the system so as to guarantee a fixed rate $ r_k$ at UE $k$ at each time slot for $k=1,2,\ldots,K$.  Imposing $R_k(t) = r_k$ into \eqref{r_k} yields
\begin{equation}\label{SINR_k_constraint}
{\rm{SINR}}_k(t) = \frac{\left|\mathbf{h}_k^H(t)\mathbf{v}_k(t)\right|^2}{\sum\limits_{i=1,i\ne k}^K\left|\mathbf{h}_k^H(t)\mathbf{v}_i(t)\right|^2 + {\sigma^2}} =\gamma_k
\end{equation}
where $\gamma_k$ is the target SINR of UE $k$ obtained as $\gamma_k = 2^{r_k}-1$. The transmitted power $P(t) = \mathsf{E}_{\mathbf{s}}[||\mathbf{V}(t)\mathbf{s}(t)||^2]$ at time $t$ is given by
\begin{align}\label{P_T}
P(t)={\rm{tr}}\left(\mathbf{V}(t)\mathbf{V}^H(t)\right)
\end{align}
while the energy consumption $E_T$ within a given time interval $[0,T]$ is obtained as
\begin{align}\label{E_T}
E_T=\int_0^T {P(t)dt} = \int_0^T {\rm{tr}}\left(\mathbf{V}(t)\mathbf{V}^H(t)\right)d t.
\end{align}
Since $\mathbf{V}(t)$ depends on the realizations of $\mathbf{W}(t)=[\mathbf{w}_1(t),\mathbf{w}_2(t),\ldots,\mathbf{w}_K(t)] \in \mathbb{C}^{N\times K}$ as well as on the user positions $\{\mathbf{x}_{k}(t);k=1,2,\ldots,K\}$, the energy $E_T$ in \eqref{E_T} is clearly a random and time-dependent quantity. To characterize its statistics, we exploit the large system limit in which $K,N \rightarrow \infty$ with fixed ratio $c =K/N$. This will allow us to measure the energy consumption of the system when $K$ is relatively large but also to capture the temporal correlations induced on the energy by UE mobility and its corresponding asymptotic distribution. The asymptotic analysis will be conducted for most of the common linear precoding techniques. We begin with the precoder that minimizes the power $P(t)$ while satisfying the rate requirements. Then, we concentrate on the asymptotically design of MRT, ZF and RZF. 

\begin{remark}A known problem with using the asymptotic analysis is that the target rates are not guaranteed to be achieved when $N$ is finite and relatively small (see for example \cite{Asgharimoghaddam14}). This is because the approximation errors are translated into fluctuations in the resulting SINR values. However, these errors vanish rapidly also in the finite regime when $N$ is large enough, which is the regime envisioned for massive MIMO systems \cite{Rusek2013a}. 
\end{remark}

\section{Optimal Linear Precoding via Large System Analysis}\label{Optimal_precoding}
We assume that at each time slot $t$ the precoding matrix $\mathbf{V}(t)$ is designed as the solution of the following power minimization problem:
\begin{align}\label{Opt_problem}
\min_{\mathbf{V}(t)}& \quad  P(t) = \tr \left( \mathbf{V}(t)\mathbf{V}^H(t) \right) \\ \label{SINR_constraints}
{\text{subject to}} &  \quad {\rm{SINR}}_k(t) \ge \gamma_k \quad k=1,2,\ldots,K.
\end{align}
As shown in \cite{Wiesel06}, the above optimization problem is not convex but it can be put in a convex form by reformulating the SINR constraints as second-order cone constraints. The optimal linear precoder (OLP) $\mathbf{V}_\text{OLP}(t)$ is eventually found to be \cite{Wiesel06,Bjornson2012a}:
\begin{align}\label{22}
\mathbf{V}_{\text{OLP}}(t)=\left(\sum\limits_{i=1}^K\lambda_i^\star(t)\mathbf{  h}_i(t)\mathbf{  h}_i^H(t) + N\mathbf{I}_N\right)^{-1}\!\!\!\!\!\mathbf{  H}(t) \sqrt{\mathbf{P}^{\star}(t)}
\end{align}
where $\mathbf{H}(t)=[\mathbf{  h}_1(t),\mathbf{  h}_2(t),\ldots,\mathbf{  h}_K(t)] \in \mathbb{C}^{N\times K}$ and $\boldsymbol{\lambda}^\star(t)= [\lambda_1^\star(t),\lambda_2^\star(t),\ldots,\lambda_K^\star(t)]^T$ is the positive unique fixed point of the following equations for $k=1,2,\ldots,K$ \cite{Wiesel06,Bjornson:2014mag}:
\begin{align}\nonumber
\left(1 + \frac{1}{\gamma_k}\right)& \lambda_k^\star(t)= \\ &\!\!\!\!\!\frac{1}{\mathbf{h}_k^H(t)\left(\sum\limits_{i=1}^K\lambda_i^\star(t)\mathbf{  h}_i(t)\mathbf{  h}_i^H(t) + N\mathbf{I}_N\right)^{-1}\!\!\!\!\!\mathbf{h}_k(t)}.\label{23}
\end{align} 
Also, $\mathbf{P}^{\star}(t)=\diag\{p_1^\star(t),p_2^\star(t),\ldots,p_K^\star(t)\}$ is a diagonal matrix whose entries are such that the SINR constraints are all satisfied with equality when $\mathbf{V} (t)= \mathbf{V}_{\text{OLP}}(t)$. Plugging \eqref{SINR_k} into \eqref{SINR_constraints}, the optimal $\mathbf{p}^\star (t)= [p_1^\star(t),p_2^\star(t),\ldots,p_K^\star(t)]^T$ is obtained as 
\begin{align}\label{24}
\mathbf{p}^\star(t) = \sigma^2 \mathbf{D}^{-1}(t) \mathbf{1}_{K}
\end{align}
where the $(k,i)$th element of $\mathbf{D}(t) \in \mathbb C^{K\times K}$ is \cite{Bjornson:2014mag}
\begin{align}\label{25}
\left[\mathbf{D}(t)\right]_{k,i}= \begin{cases}
\frac{1}{\gamma_k}{|\vect{h}_k^H(t) \vect{c}_k^\star(t)|^2}& \text{for} \,\,\, k = i \\
-{|\vect{h}_k^H(t) \vect{c}_{i}^\star(t)|^2}& \text{for} \,\,\, k \ne i
\end{cases}
\end{align}
with $ \vect{c}_k^\star(t)$ being the $k$th column of 
\begin{align}\label{25.1}
 \vect{C}^\star(t) = \left(\sum\limits_{i=1}^K\lambda_i^\star(t)\mathbf{  h}_i(t)\mathbf{  h}_i^H(t) + N\mathbf{I}_N\right)^{-1}\mathbf{  H}(t).
\end{align}
As seen, $\mathbf{V}_{\text{OLP}}(t)$ in \eqref{22} is parameterized by $\boldsymbol{\lambda}^\star(t)$ and $\mathbf{p}^\star(t)$, where $\boldsymbol{\lambda}^\star(t)$ needs to be evaluated by an iterative procedure due to the fixed-point equations in \eqref{23}. When $K,N \rightarrow \infty$ with $K/N=c \in (0,1]$, some recent tools in large system analysis allow us to compute the so-called deterministic equivalents (see \cite{Hachem2007_DeterministicEquivCertainFunctionalsRandomMatrices,Couillet_Book} for more details on this subject) of $\boldsymbol{\lambda}^\star(t)$ and $\mathbf{p}^\star(t)$. For later convenience, we call
\begin{equation}\label{26}
\eta = 1 - \frac{c}{K}\sum\limits_{i=1}^K \frac{\gamma_i}{1+\gamma_i}
\end{equation}
and
\begin{equation}\label{26.10}
 A(t) =  \frac{1}{K} \sum\limits_{i=1}^K \frac{\gamma_i}{ l(\mathbf{x}_{i}(t))}.
\end{equation}
The following theorem provides the asymptotic value of the solution to \eqref{Opt_problem}.

\begin{theorem} \label{lemma:asymptotic-beamforming}
If $K,N \rightarrow \infty$ with $K/N=c \in (0,1]$, then
\begin{align}\label{27}
\mathop {\max }\limits_{k = 1,2,\ldots,K} \left| \lambda_k^\star(t)  - \overline \lambda_k(t)\right| & \mathop {\longrightarrow}\limits^{a.s.} 0\\ \label{28}
\mathop {\max }\limits_{k = 1,2,\ldots,K} \left|p_k^\star(t)  - \overline p_k(t) \right|&\mathop {\longrightarrow}\limits^{a.s.} 0
\end{align}
where $\overline \lambda_k(t)$ and $ \overline p_k(t)$ are the deterministic equivalents of $\lambda_k^\star(t)$ and $ p_k^\star(t)$ at time slot $t$, respectively, and are given by
\begin{equation}\label{29}
\overline \lambda_k(t) =  \frac{\gamma_k}{l(\mathbf{x}_{k}(t))\eta}
\end{equation}
\begin{equation}\label{30}
\overline p_k(t) = \frac{\gamma_k}{l(\mathbf{x}_{k}(t))\eta^2}\left(\overline P_{\rm{OLP}}(t) + \frac{\sigma^2}{l(\mathbf{x}_{k}(t))} \left(1+\gamma_k\right)^2\right)
\end{equation}
with
\begin{equation}\label{31}
\overline P_{\rm{OLP}}(t) = \frac{c\sigma^2}{\eta}A(t)
\end{equation}
being the deterministic equivalent of the transmit power $P(t)$ in \eqref{P_T}.
\end{theorem}
\begin{IEEEproof}
Similar results have previously been derived by applying standard random matrix theory tools to the right-hand-side of \eqref{23}. {However, the application of these tools to the problem at hand is not analytically correct since the Lagrange multipliers in \eqref{23} are a function of the channel vectors $\{{\bf{h}}_i(t)\}$. To overcome this issue, we make use of the same approach adopted in \cite{Romain2014} whose main steps are sketched in Appendix A for completeness.} On the other hand,  \eqref{30} is proved using standard random matrix theory results and is omitted for space limitations.
\end{IEEEproof}

Observe that the Lagrange multiplier $\lambda_k(t)$ is known to act as a user priority parameter that implicitly determines how much interference the other UEs may cause to UE $k$ \cite{Bjornson:2014mag}. Interestingly, its asymptotic value $\overline \lambda_k(t)$ in \eqref{29} turns out to be proportional to the target SINR $\gamma_k$ and inversely proportional to $l(\mathbf{x}_{k}(t))$ such that users with weak channels have larger values. {This means that in the asymptotic regime higher priority is given to users that require high performance (i.e., $\gamma_k > 1$) and/or have weak average propagation conditions (i.e., $l(\mathbf{x}_{k}(t)) < 1$).} 

The following corollary can be easily obtained from Theorem \ref{lemma:asymptotic-beamforming} and will be useful later on.
\begin{corollary}[\!\!\cite{Zakhour2012}]\label{corollary_OPL_same_rates}
If the same target SINR is imposed for each user, i.e., 
\begin{equation}\label{gamma}
\boldsymbol{\gamma} = \gamma \boldsymbol{1}_k
\end{equation}
then $\overline \lambda_k(t)$ in \eqref{29} reduces to
\begin{equation}\label{26.2}
\overline \lambda_k(t) = \frac{\gamma}{l(\mathbf{x}_{k}(t))}\left( 1 - c \frac{\gamma}{1+\gamma}\right)^{-1}
\end{equation}
and $\overline P_{\rm{OLP}}(t)$ becomes
\begin{equation}\label{26.3}
\overline P_{\rm{OLP}}(t)  = {c\sigma^2}\left( 1 - c\frac{\gamma}{1+\gamma}\right)^{-1}A(t).
\end{equation}
\end{corollary}

\section{Heuristic Linear Precoding via Large System Analysis}\label{Heuristic_precoding}

Inspired by the optimal linear precoding in \eqref{22}, we also consider suboptimal precoding techniques that build on heuristics \cite{Bjornson:2014mag}.
To this end, we let $\mathbf{V}(t)$ take the generic form:
\begin{align}\label{22.101}
\!\!\mathbf{V}(t)&=\left(\sum\limits_{i=1}^K\alpha_i(t)\mathbf{  h}_i(t)\mathbf{  h}_i^H(t) + N\rho \mathbf{I}_N\right)^{-1}\!\!\!\!\!\mathbf{  H}(t)\sqrt{\mathbf{  P}(t)}
\end{align}
where $\boldsymbol{\alpha}(t)= [\alpha_1(t),\alpha_2(t),\ldots,\alpha_K(t)]^T$ is now a given vector with positive design parameters and $\rho > 0$ is another design parameter. Note that \eqref{22.101} is basically obtained from \eqref{22} by setting $\lambda_k(t) = {\alpha_k(t)}/{\rho}$ for all $k$.  As before, the power allocation matrix $\mathbf{P}(t)$ is computed so as to satisfy all the SINR constraints with equality in the asymptotic regime.


If $\mathbf{V}(t)$ takes the generic heuristic form in \eqref{22.101}, then for any given $\boldsymbol{\alpha}(t)$ and $\rho$ the results of Theorem 1 in \cite{Wagner12} lead to the following corollary.
\begin{corollary}[\!\!\cite{Wagner12}]\label{corollary_wagner}
If $\mathbf{V}(t)$ is defined as in \eqref{22.101} and $K,N \rightarrow \infty$ with $c \in (0,1]$, then 
\begin{align}\label{19}
 P(t)  - \overline P(t) \mathop {\longrightarrow}\limits^{a.s.} 0 \\ \label{20}
 {\rm{SINR}}_k(t)  - {\rm{\overline {SINR}}}_k(t) \mathop {\longrightarrow}\limits^{a.s.} 0
\end{align}
where $\overline P(t)$ and $ {\rm{\overline {SINR}}}_k(t)$ are given by
\begin{align}\label{A16}
&\overline P(t)= \frac{c \mu^\prime}{K}\sum\limits_{i=1}^K\frac{p_i(t)l(\mathbf{x}_{i}(t))}{\left(1 + \alpha_i(t)l(\mathbf{x}_{i}(t))\mu\right)^2}\\ \label{A17}
&{\rm{\overline {SINR}}}_k(t)= \frac{p_k(t)l(\mathbf{x}_{k}(t))\mu^2}{\overline P(t) + \frac{\sigma^2}{l(\mathbf{x}_{k}(t))} \left(1+\alpha_k(t) l(\mathbf{x}_{k}(t))\mu\right)^2}
\end{align}
and $\mu$ is the solution to the fixed point equation
\begin{align}\label{mu}
\mu= \left(\frac{1}{N}\sum\limits_{i=1}^K\frac{ \alpha_i(t)l(\mathbf{x}_{i}(t))}{1+\alpha_i(t) l(\mathbf{x}_{i}(t)) \mu} + \rho\right)^{-1}
\end{align}
with $\mu^\prime$ in \eqref{A16} being its derivative with respect to $\rho$.
\end{corollary}
The above corollary is now easily used to compute the asymptotic power $\overline P(t)$ required to satisfy ${\rm{\overline {SINR}}}_k(t)= \gamma_k$ for $k=1,2,\ldots,K$. Setting ${\rm{\overline {SINR}}}_k(t)$ in \eqref{A17} equal to $\gamma_k$ yields
\begin{align}\label{A19}
\overline p_k(t)& = \frac{\gamma_k}{l(\mathbf{x}_{k}(t))\mu^2} {{\left(\overline P(t) + \frac{\sigma^2}{l(\mathbf{x}_{k}(t))}\left(1+\alpha_k(t)l(\mathbf{x}_{k}(t))\mu\right)^2\right)}}.
\end{align}
Plugging \eqref{A19} into \eqref{A16} and solving with respect to $\overline P(t)$ leads to the following result.
\begin{lemma}\label{lemma2}
If ${\rm{\overline {SINR}}}_k(t)= \gamma_k$ for $k=1,2,\ldots,K$, then $\overline p_k(t)$ takes the form in \eqref{A19}
and
\begin{align}\label{12}
\overline P(t)  =  \frac{ c  \sigma^2 A(t)}{1 - \mu^2F(t)-c B(t)}
\end{align}
with $A(t)$ and $\mu$ being given by \eqref{26.10} and \eqref{mu}, respectively, and
\begin{align}\label{102}
 B(t) &= \frac{1}{K}\sum\limits_{i=1}^K \frac{ \gamma_i}{\left(1 + \alpha_i(t)l(\mathbf{x}_{i}(t))\mu\right)^2} \\ \label{102} 
F(t)&= \frac{1}{N}\sum\limits_{i=1}^{K}\frac{(\alpha_i(t)l(\mathbf{x}_{i}(t)))^2}{\left(1+\alpha_i(t)l(\mathbf{x}_{i}(t))\mu\right)^2}.
\end{align}
\end{lemma}
The above results are next used to characterize the asymptotic structure of MRT, ZF and RZF.

\subsection{Maximum Ratio Transmission}

Setting $\boldsymbol{\alpha}(t)= \boldsymbol{0}_K$ and $\rho= 1$ into \eqref{22.101} leads to 
\begin{align}\label{28.10}
\mathbf{V}_{\rm{MRT}}(t)=\frac{1}{N}\mathbf{  H}(t)\sqrt{\mathbf{  P}(t)}
\end{align}
which is the well known MRT precoder. The following corollary is obtained from Lemma \ref{lemma2}.
\begin{corollary}\label{corollary_MRT}
If MRT is used, then $\overline p_k(t)$ reduces to $\overline p_k(t) =  \frac{\gamma_k}{l(\mathbf{x}_{k}(t))}{{\left(\overline P_{\rm{MRT}}(t) + \frac{\sigma^2}{l(\mathbf{x}_{k}(t))}\right)}}{}$
where
\begin{align}
\overline P_{\rm{MRT}}(t)= \frac{c\sigma^2}{1 - c\underline{\gamma}} A(t)
\end{align}
and $\underline{\gamma}$ is the average of target SINRs:
\begin{align}\label{34}
\underline{\gamma} = \frac{1}{K}\sum\limits_{i=1}^K\gamma_i.
\end{align}

\end{corollary}
Since $\overline P_{\rm{MRT}}(t)$ must be positive and finite, it follows that the asymptotic analysis can be applied to MRT only when the following condition is satisfied:
$1 - c\underline{\gamma} > 0$. If $\boldsymbol{\gamma} = \gamma \boldsymbol{1}_k$, this implies $\gamma < 1/c$ or, equivalently, $
r < \log_2\left(1 + \frac{1}{c}\right)$.
If the above condition is met, then $\overline P_{\rm{MRT}}(t)$ is positive and finite. On the other hand, $\overline P_{\rm{MRT}}(t)$ diverges to infinity when $r =\log_2\left(1 + \frac{1}{c}\right)$.
\subsection{Zero Forcing}
Setting $\boldsymbol{\alpha}(t)= \boldsymbol{1}_K$ into \eqref{22} yields $\mathbf{V}(t)= \left(\mathbf{  H}(t) \mathbf{  H}^H(t) + N\rho \mathbf{I}_N\right)^{-1}\mathbf{  H}(t)\sqrt{\mathbf{  P}(t)}$
from which using the Woodbury matrix identity and imposing $\rho= 0$ the ZF precoder is obtained 
\begin{align}\label{ZF.1}
\mathbf{V}_{\rm{ZF}}(t)= \mathbf{  H}(t)\left(\mathbf{  H}^H(t)\mathbf{  H}(t)\right)^{-1}\sqrt{\mathbf{  P}(t)}.
\end{align}
The following corollary can be obtained from Lemma \ref{lemma2}.
\begin{corollary}\label{corollary_ZF}
If ZF is used, then $\overline p_k(t)$ reduces to $\overline p_k(t) =  \gamma_k \sigma^2$
and
\begin{align}\label{ZF.3}
\overline P_{\rm{ZF}}(t)= \frac{c\sigma^2}{1 - c}A(t).
\end{align}
\end{corollary}

\begin{proof}
Setting $\boldsymbol{\alpha}= \boldsymbol{1}_K$ into \eqref{mu} and using simple mathematical arguments yields
\begin{align}\label{mu.1}
\rho\mu = 1 - c + \frac{1}{N}\sum\limits_{i=1}^K\frac{ 1}{1+l(\mathbf{x}_{i}(t)) \mu}.
\end{align}
Letting $\varphi = \rho\mu$ one gets $\varphi = 1 - c + \frac{1}{N}\sum\nolimits_{i=1}^K\frac{ 1}{1+l(\mathbf{x}_{i}(t)) \varphi/\rho}$ from which it follows that $\varphi \to 1- c$ when $\rho \to 0$. This means that $\mu$ goes to infinity as $\frac{1-c}{\rho}$ when $\rho \to 0$. Therefore, from \eqref{A19} one easily gets $\overline p_k(t) -  \gamma_k \sigma^2 \to 0$ whereas \eqref{102}  leads to $B(t) \to 0$ and $\mu^{2}F(t) \to c$. Plugging these results into \eqref{12} yields \eqref{ZF.3}.
\end{proof}

\subsection{Regularized Zero Forcing}
Assume that ${\alpha}_i(t) = 1/l(\mathbf{x}_{i}(t))$ for any $i$, then the processing matrix $\mathbf{V}(t)$ in \eqref{22} reduces to (see also \cite{Evans2013})
\begin{align}\label{31}
\!\!\mathbf{V}_{\rm{RZF}}(t)&=\left(\sum\limits_{i=1}^K\mathbf{  w}_i(t)\mathbf{  w}_i^H(t) + N\rho\mathbf{I}_N\right)^{-1}\!\!\!\!\!\mathbf{  H}(t)\sqrt{\mathbf{  P}(t)}
\end{align}
which is referred to as RZF precoder in the sequel. Differently from OLP that requires to compute the fixed point of a set of equations, the optimization of RZF requires only to look for the value of $\rho$. This can not generally be done in closed-form but requires a numerical optimization procedure \cite{Bjornson:2014mag}. If the asymptotic regime is analyzed, then the following result is obtained.
\begin{lemma}
If a RZF precoder is used and $K,N \rightarrow \infty$ with $K/N=c \in (0,1]$, then the optimal regularization parameter is found to be
\begin{align}\label{32.1}
\rho^\star = \frac{1}{\underline{\gamma}} - \frac{c}{1 +\underline{\gamma}}
\end{align}
with $\underline{\gamma}$ given by \eqref{34}. The deterministic equivalent of the transmit power reduces to
\begin{align}\label{36.10}
\overline P_{\rm{RZF}}(t)= {c\sigma^2}\left(1 - c \frac{\underline{\gamma}}{1+\underline{\gamma}}\right)^{-1} A(t)
\end{align}
whereas $\overline p_{k}(t)= \frac{\gamma_k}{l(\mathbf{x}_{k}(t))\underline{\gamma}^2} {{\left(\overline P_{\rm{RZF}}(t) + \frac{\sigma^2}{l(\mathbf{x}_{k}(t))}\left(1+\underline{\gamma}\right)^2\right)}}.$
\end{lemma}

\begin{IEEEproof}
Setting $\alpha_i =1/l(\mathbf{x}_{i}(t))$ for $i=1,2,\ldots,K$ into \eqref{mu} and \eqref{12} -- \eqref{102} yields
\begin{align}\label{C30}
\overline P_{\rm{RZF}}(t)  = c\sigma^2\frac{ \left(1+\mu\right)^2A(t)}{\mu\left(c+\rho\left(1+\mu\right)^2\right) -c\underline{\gamma}}
\end{align}
with
\begin{align}\label{mu.10}
\mu = \left(\frac{c}{1 + \mu} + \rho\right)^{-1}.
\end{align}
Taking the derivative of $\overline P_{\rm{RZF}}(t)$ in \eqref{C30} with respect to $\rho$, one gets (the mathematical steps are omitted for space limitations)
\begin{align}\label{C31}
\frac{\partial \overline P_{\rm{RZF}}(t)}{\partial \rho}  =2c^2\sigma^2\frac{\left(\underline{\gamma}- \mu\right) A(t)}{\left(\mu \left(c + \rho\left(1+\mu\right)^2\right) - c\underline{\gamma}\right)^2}.
\end{align}
From the above equation, it turns out that the minimum power is achieved when $\mu$ is such that $\mu= \underline{\gamma}$. Plugging this result into \eqref{mu.10} yields \eqref{32.1}.
\end{IEEEproof}

{To the best of our knowledge, this is the first time that the value of $\rho$ that minimizes the power consumption in the asymptotic regime is given in the explicit form of \eqref{32.1}.} Most of works have only looked for the value of $\rho$ that maximizes the sum rate of the network (see for example \cite{Wagner12}). {Interestingly, the results of Lemma 2 can be used to prove the following corollary, which states that RZF is asymptotically equivalent to the optimal linear precoder when the same rate constraint is imposed for all UEs.}
\begin{corollary}\label{optimality_RZF}
If the same target SINR is imposed for each user, i.e., $\boldsymbol{\gamma} = \gamma \mathbf{1}_K$,
then RZF becomes optimal in the asymptotic regime.
\end{corollary}
\begin{IEEEproof}If $\boldsymbol{\gamma} = \gamma \boldsymbol{1}_k$, then \eqref{32.1} reduces to $\rho^\star = \frac{1}{\gamma} - \frac{c}{1 + \gamma}$
and $\mathbf{V}_{\rm{RZF}}(t)$ in \eqref{31} becomes equivalent to \eqref{22} after replacing $\lambda_k^\star$ with $\overline\lambda_k$ given by \eqref{26.2}.
\end{IEEEproof}

Observe that if $\boldsymbol{\alpha}(t)$ is set to $\boldsymbol{1}_K$, then $\mathbf{V}(t)$ in \eqref{22} reduces to\cite{Wagner12}:
\begin{align}\label{28}
\mathbf{V}_{\rm{RZF}}(t)&=\left(\sum\limits_{i=1}^K\mathbf{  h}_i(t)\mathbf{  h}_i^H(t) + N\rho\mathbf{I}_N\right)^{-1}\!\!\!\!\!\mathbf{  H}\sqrt{\mathbf{  P}}(t).
\end{align}
As shown in \cite{Sanguinetti2014_Globecom}, if $K,N \rightarrow \infty$ with $c \in (0,1]$ and $\mathbf{V}_{\rm{RZF}}(t)$ is given as above, then the value of $\rho$ minimizing $P(t) = \tr \left( \mathbf{V}(t)\mathbf{V}^H(t) \right)$ while satisfying the rate constraints is obtained as
\begin{align}\label{29.1}
\rho^\star = \frac{1}{\mu^\star} - \frac{1}{N}\sum\limits_{i=1}^K\frac{l(\mathbf{x}_{i}(t))}{1+l(\mathbf{x}_{i}(t)) \mu^\star}
\end{align}
with $\mu^\star$ being solution of the following fixed point equation:
\begin{align}\label{30.1}
\mu^\star = \left(\sum\limits_{i=1}^K \frac{l(\mathbf{x}_{i}(t)) \gamma_i}{\left(1 + l(\mathbf{x}_{i}(t))\mu^\star\right)^3}\right)\left(\sum\limits_{i=1}^K \frac{\left(l(\mathbf{x}_{i}(t))\right)^2}{\left(1 + l(\mathbf{x}_{i}(t))\mu^\star\right)^3}\right)^{-1}\!\!\!\!\!.
\end{align}
In addition, it turns out that \cite{Sanguinetti2014_Globecom}
\begin{align}\label{36.100}
\overline P_{\text{RZF}}(t)= \frac{c\sigma^2A(t)}{1 - (\mu^\star)^2 F(t) - cB(t)}
\end{align}
where $A(t)$ is given by \eqref{26.10} while $B(t)$ and $F(t)$ are obtained from \eqref{102} setting ${\alpha}_i(t)=1$ $\forall i$ and replacing $\mu$ with $\mu^\star$.
The above results are used in \cite{Sanguinetti2014_Globecom} to prove that the conventional RZF processing given by \eqref{28} becomes optimal in the asymptotic regime only if the ratio between $\gamma_k$ and $l(\mathbf{x}_{k}(t))$ is the same for any $k$. Although possible, this assumption is not very realistic in practical systems as it would imply a strong dependence between rate requirements and UE positions. For this reason, we focus our analysis on the RZF given by \eqref{31} whose optimality holds true under a more reasonable assumption of identical target UE rates.

\section{Energy Fluctuations}\label{Energy_fluctuations}
We are now left with studying the fluctuations of $E_T$ in the large system limit. Towards this end, observe that the large system analysis above shows that, for all considered schemes,  the transmit power $P(t)= \tr \left( \mathbf{V}(t)\mathbf{V}^H(t) \right)$ hardens to a deterministic quantity, which takes the following general form:
\begin{align}\label{1010}
\overline P(t) = \frac{c\sigma^2}{\eta} A(t) = \frac{c\sigma^2}{\eta} \frac{1}{K} \sum\limits_{i=1}^K \frac{\gamma_i}{l({\bf x}_i(t))}
\end{align}
with $\eta$ being specifically given by
\begin{align}\label{eta_different_linear}
\eta = \left\{ {\begin{array}{*{20}{l}}
{1 - \frac{c}{K}\sum\limits_{i=1}^K \frac{\gamma_i}{1+\gamma_i}}&{\text{for OLP}}\\
{1 - c\underline{\gamma} = 1 - \frac{c}{K}\sum\limits_{i=1}^K\gamma_i}&{\text{for MRT}}\\
{1-c = 1-\frac{K}{N}}&{\text{for ZF}}\\
{1 - c \frac{\underline{\gamma}}{1+\underline{\gamma}} = 1 -c \frac{\sum\limits_{i=1}^K\gamma_i}{K+\sum\limits_{i=1}^K\gamma_i}}&{\text{for RZF.}}
\end{array}} \right.
\end{align}
As seen, in the large system limit the power expenditure of all considered schemes is a function of time due to the fluctuations induced on the average channel attenuation by $\{{\bf x}_i(t);i=1,2,\ldots,K\}$. The statistics of the large-scale UE movements can be exploited to calculate the temporal fluctuations of each term in the right-hand-side of \eqref{1010}. These results are stated in the following two lemmas, which are proved in Appendix C.
\begin{lemma}\label{lemma_mean}
The average value of $l^{-1}({\bf x}_k(t))$ does not depend on the UE movements and hardens to a deterministic quantity independent of time $t$, which is the same for all UEs and given by
\begin{align}\label{100}
\mathsf{E}_{\mathbf{X}_k}\left[l^{-1}({\bf x}_k(t))\right] = \frac{1}{|\mathcal C|}\int_{\mathcal C} \frac{1}{l(\mathbf{x})}d{\bf x}
\end{align}
with $|\mathcal C| = \pi R^2$.
\end{lemma}
\begin{lemma} \label{lemma_covariance}The covariance of $l^{-1}({\bf x}_k(t))$ depends on the temporal correlation induced by UE movements and is such that 
\begin{align}\label{101}
\iint_{0}^T{\mathsf{COV}}_{\mathbf{X}_k}\left[l^{-1}({\bf x}_k(\tau)),l^{-1}({\bf x}_k(s))\right] d\tau ds= \frac{T R^{2}}{D} \Theta
\end{align}
where \begin{align}\label{105}
\Theta &= \sum\limits_{i=1}^\infty \frac{2\phi_{i}^2}{ \kappa_{i}^2 J_0^2(\kappa_{i})}\int_0^1{\left(1-e^{-{\frac{\kappa_{i}^2DTt}{R^2}}}\right)^2 dt}
\end{align}
and
\begin{align}\label{106}
\phi_{i} =2\int_0^1  \frac{1} {{l(R\mathbf{z})}} J_0(\kappa_{i}z) z dz
\end{align}
with $\kappa_{i}$ being the $i$th zero of  $J_1(x)$.
\end{lemma}
The values of $\{\kappa_{i}\}$ in \eqref{105} and \eqref{106} can be found in \cite[pp. 390] {Abramowitz_Stegun_book} whereas $\{\phi_i\}$ can be calculated explicitly using formulae in \cite[pp. 684]{Gradshteyn_Ryzhik_book}. Observe also that the sum in \eqref{105} requires only to compute a few terms as it converges fast to its effective value. 

The following theorem summarizes one of the major results of this work.
\begin{theorem}\label{CLT}
If linear precoding is used at the BS, then the following convergence holds true in the large system limit:
\begin{align}\label{convergence}
\sqrt{K}\left(\frac{E_T - \epsilon}{\sqrt{\Sigma}}\right) \mathop  {\longrightarrow} \limits_{K,N \to \infty}^{\mathcal D} \mathcal{N}(0,1)
\end{align}
where
\begin{align}\label{avg_pow}
\epsilon = T\frac{c\sigma^2}{\eta} \left(\frac{1}{K}\sum\limits_{i=1}^K\gamma_i\right) \frac{1}{|\mathcal C|}\int_{\mathcal C} \frac{1}{l(\mathbf{x})}d{\bf x}
\end{align}
\begin{align}\label{var_pow}
\Sigma &= \left(\frac{c\sigma^2}{\eta}\right)^2\left(\frac{1}{K}\sum\limits_{i=1}^K\gamma_i^2\right) \frac{T R^{2}}{D} \Theta
\end{align}
with $\Theta$ being computed as in \eqref{105}.
\end{theorem}
\begin{IEEEproof}
The proof is given in Appendix D.
\end{IEEEproof}

{{\begin{table*}[t]
\renewcommand{\arraystretch}{1.}
\caption{Parameter setting}
\label{table_coefficients}
\centering
\begin{tabular}{|c|c||c|c|}
\hline
{\bf Parameter} &  {\bf Value} & {\bf Parameter} &  {\bf Value}\\
\hline
Bandwidth &  $W =20$ MHz & Diffusion coefficient & $D = 1250$ m$^2$/minute
\\
\hline
Cell radius &  $R=500$ m & Pathloss coefficient & $ \beta=4$ \\
\hline
Cut-off parameter &  $\bar x =  25$ m & Average pathloss attenuation at $\bar x$ & $ L_{\bar x} = -93$ dB\\
\hline
Carrier frequency & $f_c = 2.4$ GHz & Noise power & $\sigma^2 =-97.8$ dBm \\
\hline
Average step size length &  $\ell =50$ m & Time interval & $T = 3,6,12$ or $24$ hrs\\
\hline
\end{tabular}
\end{table*}}}

\begin{table*}[t]
\renewcommand{\arraystretch}{1.}
\caption{{Numerical values of $\mathsf{E}\left[E_T\right]/\epsilon$ and $K \mathsf{VAR}\left[E_T\right]/\Sigma$ for the parameter setting in Figs. \ref{fig1_A} -- \ref{fig3_A}.} when $T=3,6,12$ and $24$ [hrs].}
\label{table_values} \vskip2mm
\centering
\begin{tabular}{|l||c|c||c|c||c|c|}
\hline
&  {{Fig. 1 }}& {{Fig. 1 }}& {{Fig. 2 }}& {{Fig. 2 }}& {{Fig. 3 }}& {{Fig. 3 }}\\
\hline
&  $\mathsf{E}\left[E_T\right]/\epsilon$ &  $K\mathsf{VAR}\left[E_T\right]/\Sigma$ & $\mathsf{E}\left[E_T\right]/\epsilon$ &  $K\mathsf{VAR}\left[E_T\right]/\Sigma$ & $\mathsf{E}\left[E_T\right]/\epsilon$ &  $K\mathsf{VAR}\left[E_T\right]/\Sigma$\\
\hline
$T=3$ [hrs] & $1.030$  & $1.1887 $ & $1.023  $  & $1.1207$  & $1.019$  & $1.0675$\\
$T=6$ [hrs] & $1.011$  & $1.1287 $ & $1.008  $  & $1.0901$  & $1.006$  & $1.048$\\

$T=12$ [hrs] & $1.005$  & $1.0659$ & $1.004$  & $1.0428$  & $1.003$  & $1.021$\\

$T=24$ [hrs] & $1.001$  & $1.0459$ & $1.000$  & $1.0211$  & $1.000$  & $1.018$\\
\hline
\end{tabular}
\end{table*}

Theorem \ref{CLT} says that the fluctuation of $E_T$ around $\epsilon$ is approximately Gaussian with variance $\Sigma/K$. Some intuition on this result is as follows. As mentioned earlier, the temporal variability of the energy $E_T$ basically depends through $\mathbf{V}(t)$ on the variations of the small-scale fading matrix $\{\mathbf{W}(t)\}$ and the UE positions $\{{{\mathbf{x}}_i(t)};i=1,2,\ldots,K\}$ in the cell. However, the statistics of these two contributions and their impact on the fluctuations of $E_T$ are relatively different. To see how this comes about, recall that the coherence time $\Delta\tau$ of $\mathbf{W}(t)$ is $\sim \lambda\xi/\ell$ where $\lambda$ is the wavelength and {$\ell/\xi$} is the UE velocity between two successive steps. In addition, the variance of the power required to only compensate $\mathbf{W}(t)$ (which corresponds to the case of motionless UEs) is known to scale as $1/K^{2}$ due to the $K^2$ degrees of freedom of $\mathbf{W}(t)$ \cite{Bai2004_CLT_covariance_matrices}. This means that the energy fluctuations induced by the compensation of the small-scale fading only behaves roughly as $\sim \Delta\tau/K^2$. {Differently from the small-scale fading (where the coherence distance is $\sim \lambda\xi$), the coherence distance of UE movements (defined as the distance after which the correlation between the initial and final position becomes small) is related to the cell radius $R$ and the corresponding coherence time can be roughly quantified as $\sim R^2/D$ with $4D = \ell^2/\xi$. The latter is much larger than $\Delta\tau \sim \lambda\xi/\ell$ for typical parameter settings (see Example \ref{example_2}).
In addition, the variance of the power due to the randomness of $\{{\mathbf{x}_i(t)}\}$ scales as $1/K$ rather than as $1/K^2$ (see Appendix D). Putting these two facts together, it follows that the energy fluctuations induced by UE mobility are on the order of $\sim R^2/(KD)$ and thus largely dominate the variability of the small-scale fading, which is on the order of $\sim \Delta\tau/K^2$.}
\begin{example}
Assume that the pathloss function $l(\mathbf{x})$ is modelled as in \eqref{avg_pathloss_1}. In these circumstances, 
\begin{align}\label{avg_pow_2}
\epsilon = T\frac{c\sigma^2}{\eta} \left(\frac{1}{K}\sum\limits_{i=1}^K\gamma_i\right) \frac{R^{\beta}}{2\bar x^{\beta} L_{\bar x}}\left(\frac{2}{2+\beta} + \frac{\bar x^{\beta}}{R^{\beta}}\right)
\end{align}
and the coefficients $\{\phi_i\}$ in \eqref{var_pow} take the form:\footnote{Observe that $\int_0^1 J_0(\kappa_{i}x) x dx = 0$.}
\begin{align}\label{106.10}
\phi_{i} =\frac{R^{\beta}}{\bar x^{\beta} L_{\bar x}}\int_0^1 x {^\beta} J_0(\kappa_{i}x) x dx.
\end{align}
If $\beta= 4$, using the formulae in \cite[pp. 684]{Gradshteyn_Ryzhik_book} one gets
\begin{align}
\int_0^1 x {^\beta} J_0(\kappa_{i}x) x dx=
4J_0(\kappa_i)\frac{\kappa_i^2-8}{\kappa_i^4}.
\end{align}
Plugging the above results into \eqref{105} yields $\Theta = \Omega{R^{2\beta}}$
with
\begin{align}\label{105.11}
\Omega &= \frac{1}{\bar x^{2\beta} L^2_{\bar x}} \sum\limits_{i=1}^\infty \frac{32(\kappa_i^2-8)^2}{\kappa_i^{10}} \int_0^1{\left(1-e^{-{\frac{\kappa_{i}^2DTt}{R^2}}}\right)^2 dt}.
\end{align}
The variance $\Sigma$ in \eqref{var_pow} is eventually obtained as 
\begin{align}\label{var_pow_2}
\Sigma &= \left(\frac{c\sigma^2}{\eta}\right)^2\left(\frac{1}{K}\sum\limits_{i=1}^K\gamma_i^2\right) \Omega \frac{T R^{2}}{D} R^{2\beta}.
\end{align}
From the closed-form expressions in \eqref{avg_pow_2} and \eqref{var_pow_2}, it follows that the mean and variance of $E_T$ depend heavily on the values of the path-loss exponent $\beta$ and the cell radius $R$. 
\end{example} 
\begin{remark}
From the results of Theorem \ref{CLT}, it follows that the mean and variance of $E_T$ are both proportional to the time interval $T$. This means that the variability of $E_T/T$ will be less important as $T$ becomes large. Observe also that if $T$ is such that $DT/R^2\gg 1$ then the integral in \eqref{105} converges to unity and the ratio $\Sigma/T$ becomes independent from $T$. This large time limiting behavior is complementary to the large $K$ and $N$ analysis performed in this work. We will elaborate further on this result later on in Section \ref{Conclusions} as it can be exploited to consider an alternative regime in which $K$ and $N$ are finite while $T$ grows large. 
\end{remark} 

\begin{remark}
It is worth observing that the results of Theorem \ref{CLT} can in principle be extended to other linear pre-coding techniques in which $\overline P(t)$ has a more involved structure. In particular, it can be extended to the classical RZF precoder whose asymptotic power $\overline P_{RZF}(t)$ is given in \eqref{36.100}. Although being in a more complicated form than \eqref{1010}, the mean and variance of $E_T$ are still computable as they basically require to evaluate the fluctuations induced by UE movements on the different terms $\mu^\star$, $A(t)$, $B(t)$ and $F(t)$. 
\end{remark}

\section{Numerical validation and Applications}\label{Numerics}
\subsection{Numerical validation}
The accuracy of the above asymptotic statistical characterization is now validated numerically by Monte-Carlo simulations. The results are obtained for 1000 different initial positions ${\bf x}_k (0)$ for $k=1,2,\ldots,K$ within the coverage area. The simulations were performed using Matlab and the code is available for download\footnote{https://github.com/lucasanguinetti/energy-consumption-in-MU-MIMO-with-mobility.}, which enables reproducibility as well as simple testing of other parameter values. The parameter setting is given in Table \ref{table_coefficients}. The pathloss function $l(\mathbf{x})$ is modelled as in \eqref{avg_pathloss_1} with $\beta = 4$ and $L_{\bar x}=-93$ dB. The latter is such that for $f_c=2.4$ GHz the attenuation at $\bar x$ is the same as that in the cellular model analyzed in \cite{Calcev2007}.

Figs.~\ref{fig1_A} -- \ref{fig3_A} show the cumulative distribution function (CDF) of ${E_T}/{T}$ for different values of $K$, $N$ and $T$. The ratio $c=K/N$ is kept fixed and equal to $c = 0.5$ whereas the user rates $r_k$ are all set to $1.5$ bit/s/Hz. As seen, the simulation results match pretty well with the theoretical ones for both OLP/RZF and ZF in all investigated scenarios.\footnote{Recall that RZF becomes equivalent to OLP when the user rates are all the same (see Corollary \ref{optimality_RZF}).} This validates the theoretical analysis of this work and shows a substantial energy reduction in using OLP. The CDF of ${E_T}/{T}$ for MRT has the same behaviour of CDFs of OLP and RZF and it is not reported for illustration purposes as the average power required by MRT is larger than ZF and OLP (or RZF) and thus it would have comprised the legibility of the results. In particular, it turns out that if MRT is used then the mean value of ${E_T}/{T}$ is increased by a factor of $7.92$ with respect to OLP. This is in agreement to what can be easily obtained computing the ratio between the two different values of \eqref{eta_different_linear} for OLP and MRT. A close inspection of Figs.~\ref{fig1_A} -- \ref{fig3_A} reveals that for a given $T$ there is a progressive agreement between numerical and theoretical results as $K$ increases (see also the numerical values reported in Table \ref{table_values}). This is due to the fact that the finite size effect of $T$ becomes less relevant as $K$ grows large. A similar behaviour is observed if one compares the results of Figs.~\ref{fig1_A}--\ref{fig3_A} fixing $K$ while letting $T$ increase. This is because increasing $T$ allows each UE to cover larger areas of the cell and, as a consequence, the energy consumption becomes less random and closer to the mean. This in turn reduces the corresponding fluctuations.

\begin{figure}[t!]
\begin{center}
    \psfrag{xlabel}[c][b]{{\footnotesize{$\alpha$ [Watt]}}}
    \psfrag{ylabel}[c][t]{\footnotesize{$\text{Pr}\left(\frac{E_T}{T} > \alpha \right)$}}
    \psfrag{X1}[c][m]{\scriptsize{\text{OLP/RZF \quad\;}}}
    \psfrag{x2}[c][m]{\scriptsize{\text{ZF}}}
    \psfrag{data1}[l][m]{\!\!\tiny{T = 12 [h] -- Sim.}}
    \psfrag{data2}[l][m]{\!\!\tiny{T = 12 [h] -- Theory}}
        \psfrag{data3}[l][m]{\!\!\tiny{T = 3 [h] -- Sim.}}
    \psfrag{data4}[l][m]{\!\!\tiny{T = 3 [h] -- Theory}}
        \psfrag{title}[c][m]{\!\!\!\!\footnotesize{$K= 16, N=32$}}
\includegraphics[width=8.2cm]{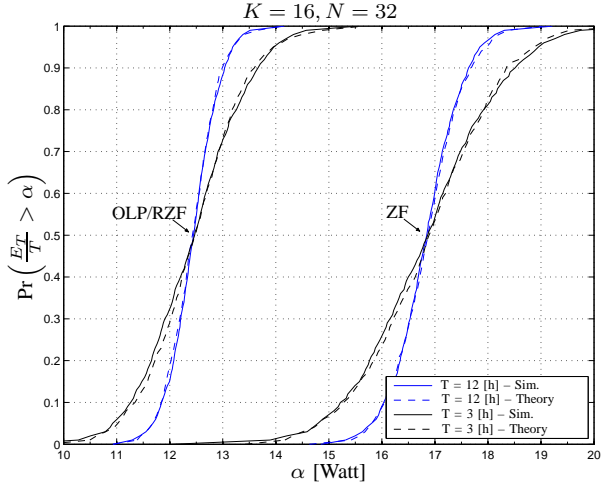}
\caption{\footnotesize{Outage probability $\text{Pr}\left(\frac{E_T}{T}> \alpha \right)$ of OLP and ZF when $K=16$ and $T=3$ or 12 hrs.}} \label{fig1_A}
\end{center}
\end{figure}

\begin{figure}
\begin{center}
    \psfrag{xlabel}[c][b]{{\footnotesize{$\alpha$ [Watt]}}}
    \psfrag{ylabel}[c][t]{\footnotesize{$\text{Pr}\left(\frac{E_T}{T} > \alpha \right)$}}
    \psfrag{X1}[c][m]{\scriptsize{\text{OLP/RZF \quad\;}}}
    \psfrag{X2}[c][m]{\scriptsize{\text{ZF}}}
    \psfrag{data1}[l][m]{\!\!\tiny{T = 12 [h] -- Sim.}}
    \psfrag{data2}[l][m]{\!\!\tiny{T = 12 [h] -- Theory}}
        \psfrag{data3}[l][m]{\!\!\tiny{T = 3 [h] -- Sim.}}
    \psfrag{data4}[l][m]{\!\!\tiny{T = 3 [h] -- Theory}}
        \psfrag{title}[c][m]{\!\!\!\!\footnotesize{$K= 32, N=64$}}
\includegraphics[width=8.2cm]{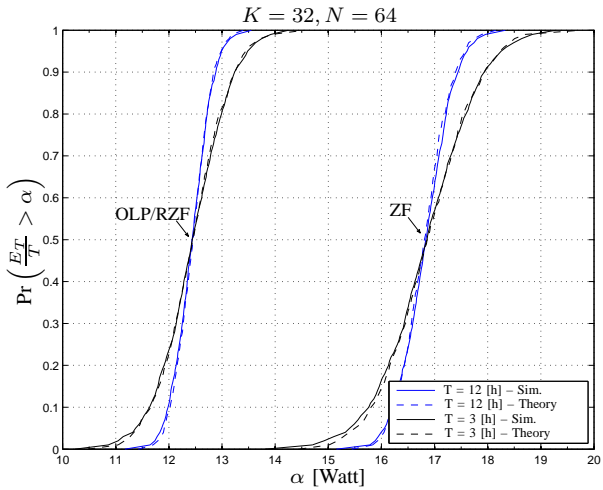}
\caption{\footnotesize{Outage probability $\text{Pr}\left(\frac{E_T}{T}> \alpha \right)$ of OLP and ZF when $K=32$ and $T=3$ or 12 hrs.}} \label{fig2_A}
\end{center} 
 \end{figure}

To evaluate the impact of the pathloss coefficient on the energy consumption statistics, we set $\beta=6$ (as an extreme case). In these circumstances, the integral on the right-hand-side of \eqref{106.10} can be computed as \cite{Gradshteyn_Ryzhik_book}: 
\begin{align}
\int_0^1 x {^\beta} J_0(\kappa_{i}x) x dx= 6J_0(\kappa_i)\frac{\kappa_i^4-24\kappa_i^2+192}{\kappa_i^6}.
\end{align}
Fig.~\ref{fig4_A} show the CDF of ${E_T}/{T}$ for $\beta=6$ when $K=32,N=64$ and $r_k=1,5$ bit/s/Hz. Comparing the results of Fig.~\ref{fig4_A} with those in Fig. \ref{fig2_A}, it is seen that both the mean and variance values of ${E_T}/{T}$ largely increase as $\beta$ becomes larger.

   \begin{figure}[t!]
\begin{center}
    \psfrag{xlabel}[c][b]{{\footnotesize{$\alpha$ [Watt]}}}
    \psfrag{ylabel}[c][t]{\footnotesize{$\text{Pr}\left(\frac{E_T}{T} > \alpha \right)$}}
    \psfrag{X1}[c][m]{\scriptsize{\text{OLP/RZF \quad\;}}}
    \psfrag{X2}[c][m]{\scriptsize{\text{ZF}}}
    \psfrag{data1}[l][m]{\!\!\tiny{T = 12 [h] -- Sim.}}
    \psfrag{data2}[l][m]{\!\!\tiny{T = 12 [h] -- Theory}}
        \psfrag{data3}[l][m]{\!\!\tiny{T = 3 [h] -- Sim.}}
    \psfrag{data4}[l][m]{\!\!\tiny{T = 3 [h] -- Theory}}
        \psfrag{title}[c][m]{\!\!\!\!\footnotesize{$K= 64, N=128$}}
\includegraphics[width=8.2cm]{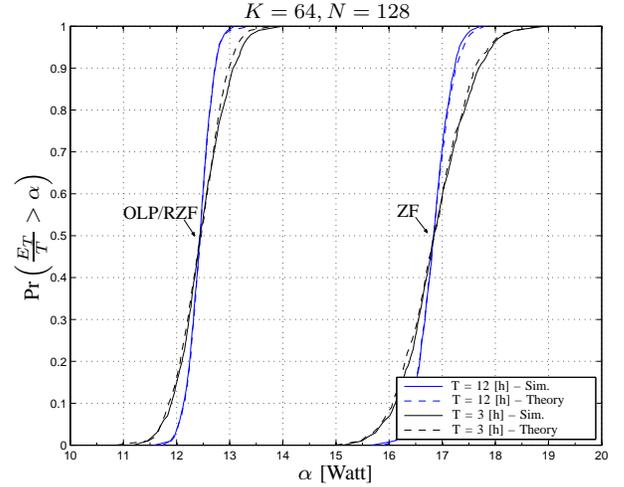}
\caption{\footnotesize{Outage probability $\text{Pr}\left(\frac{E_T}{T}> \alpha \right)$ of OLP and ZF when $K=64$ and $T=3$ or 12 hrs.}} \label{fig3_A}
\end{center} 
\end{figure}

\begin{figure}
\begin{center}
    \psfrag{xlabel}[c][b]{{\footnotesize{$\alpha$ [Watt]}}}
    \psfrag{ylabel}[c][t]{\footnotesize{$\text{Pr}\left(\frac{E_T}{T} > \alpha \right)$}}
    \psfrag{X1}[c][m]{\scriptsize{\text{OLP/RZF \quad\;}}}
    \psfrag{X2}[c][m]{\scriptsize{\text{ZF}}}
    \psfrag{data1}[l][m]{\!\!\tiny{T = 12 [h] -- Sim.}}
    \psfrag{data2}[l][m]{\!\!\tiny{T = 12 [h] -- Theory}}
        \psfrag{data3}[l][m]{\!\!\tiny{T = 3 [h] -- Sim.}}
    \psfrag{data4}[l][m]{\!\!\tiny{T = 3 [h] -- Theory}}
        \psfrag{title}[c][m]{\!\!\!\!\footnotesize{$K= 32, N=64, \beta =6$}}
\includegraphics[width=8.2cm]{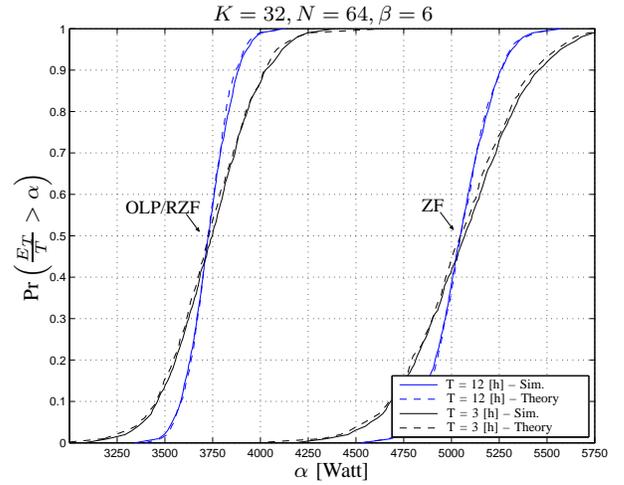}
\caption{\footnotesize{Outage probability $\text{Pr}\left(\frac{E_T}{T} > \alpha \right)$ of OLP and ZF when $K=32$, $\beta=6$ and $T=3$ or 12 hrs.}} \label{fig4_A}
\end{center}
 \end{figure}

\subsection{Dimensioning of cell battery}\label{Applications_battery}
A possible application of the results of Theorem \ref{CLT} is as follows. Assume that the energy level $\mathcal E$ of a battery-powered BS has to be designed such that a certain rate is guaranteed at each UE and the probability of running out of energy (before replacement or reloading) is smaller than some given threshold $\chi$. Mathematically, this amounts to saying that $\text{Pr}\left({E_T} > \mathcal E \right) \le \chi$. From the results of Theorem \ref{CLT}, one gets
\begin{align}\label{outage_1}
\text{Pr}\left({E_T} > \mathcal E \right) = Q\left(\sqrt{K}\,\frac{\mathcal E - \epsilon}{\sqrt{\Sigma}}\right)
\end{align}
from which it follows that
\begin{align}\label{outage_1}
\mathcal E \ge \frac{\sqrt{\Sigma}}{\sqrt K} Q^{-1}\left(\chi\right) + \epsilon.
\end{align}
Fig. \ref{fig8} illustrates the battery level $\mathcal E$ as a function of rate for different values of $K$ when $N=128,\beta = 4 $ and the replacing (or recharging) time $T$ is $12$ hrs. {Marks indicate simulation results while solid lines are obtained theoretically.} Fig. \ref{fig9} reports the values of $\mathcal E$ as a function of $K$ in the same operating conditions of Fig. \ref{fig8} with the only difference that now the UE rate requirements take values within the interval $[1, 4]$ bit/s/Hz. From the results of Figs. \ref{fig8} and \ref{fig9}, it follows that OLP and RZF provide a substantial energy saving with respect to ZF when $K$ increases. As expected, the saving is more relevant for moderate values of rate requirements $r$ since ZF is known to be suboptimal in that regime. From Fig. \ref{fig8}, it follows that for low data rates (in the range of $r=0.5$ bit/s/Hz) MRT requires the same battery level of OLP while a substantial increase is observed as $r$ grows up to $\log_2\left(1 + 1/c\right)$.

 \begin{figure}[t!]
\begin{center}
    \psfrag{xlabel}[c][b]{{\footnotesize{$r$ [bit/s/Hz]}}}
    \psfrag{ylabel}[c][t]{\footnotesize{Battery level $\mathcal E$ [Watt$\cdot$hr]}}
    \psfrag{X1}[c][m]{\tiny{$K =16$}}
    \psfrag{X2}[c][m]{\tiny{$K =64$}}
        \psfrag{X3}[c][m]{\tiny{$K =112$}}
    \psfrag{data1}[l][m]{\!\!\tiny{OLP}}
    \psfrag{data2}[l][m]{\!\!\tiny{ZF}}
        \psfrag{data3}[l][m]{\!\!\tiny{MRT}}
        \psfrag{title}[c][m]{\!\!\!\!\footnotesize{$N=128, T = 12$ hrs}}
\includegraphics[width=7.8cm]{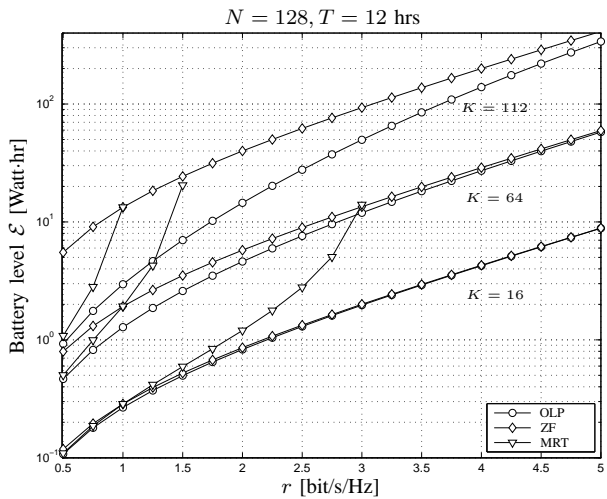}
\caption{\footnotesize{Battery level that is required by OLP, ZF and MRT to meet $\text{Pr}\left({E_T} > \mathcal E \right) \le 1\%$ when $K=16,64$ or $112$.}} \label{fig8}
\end{center}
\end{figure}

\begin{figure}
\begin{center}
    \psfrag{xlabel}[c][b]{{\footnotesize{Number of users, $K$}}}
    \psfrag{ylabel}[c][t]{\footnotesize{Battery level $\mathcal E$ [Watt$\cdot$hr]}}
    \psfrag{data1}[l][m]{\!\!\tiny{OLP}}
    \psfrag{data2}[l][m]{\!\!\tiny{RZF}}
        \psfrag{data3}[l][m]{\!\!\tiny{ZF}}
        \psfrag{title}[c][m]{\!\!\!\!\footnotesize{$N=128,T = 12$ hrs}}
\includegraphics[width=7.8cm]{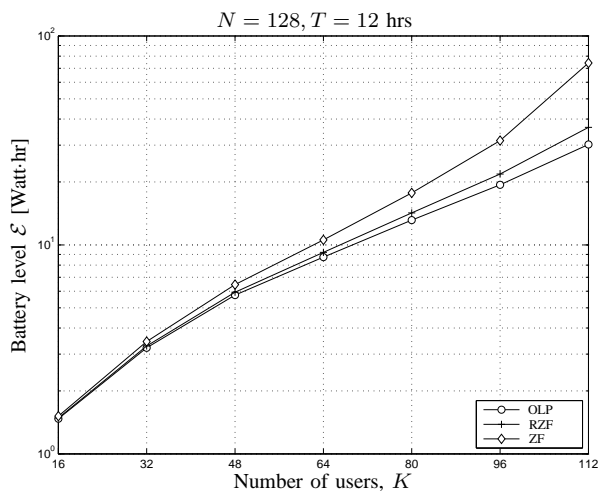}
\caption{\footnotesize{Battery level that is required by OLP, RZF and ZF to meet $\text{Pr}\left({E_T} > \mathcal E \right) \le 1\%$ when $r_k \in [1, 4]$ bit/s/Hz.}} \label{fig9}
\end{center}
 \end{figure}
 
   \begin{figure}[t!]
\begin{center}
    \psfrag{xlabel}[c][b]{{\footnotesize{$r$ [bit/s/Hz]}}}
    \psfrag{ylabel}[c][t]{\footnotesize{Cell radius $R$ [meter]}}
    \psfrag{X1}[c][m]{\tiny{$K =16$}}
    \psfrag{X2}[c][m]{\tiny{$K =64$}}
        \psfrag{X3}[c][m]{\tiny{$K =112$}}
    \psfrag{data1}[l][m]{\!\!\tiny{OLP}}
    \psfrag{data2}[l][m]{\!\!\tiny{ZF}}
        \psfrag{data3}[l][m]{\!\!\tiny{MRT}}
        \psfrag{title}[c][m]{\!\!\!\!\footnotesize{$N=128$}}
\includegraphics[width=7.8cm]{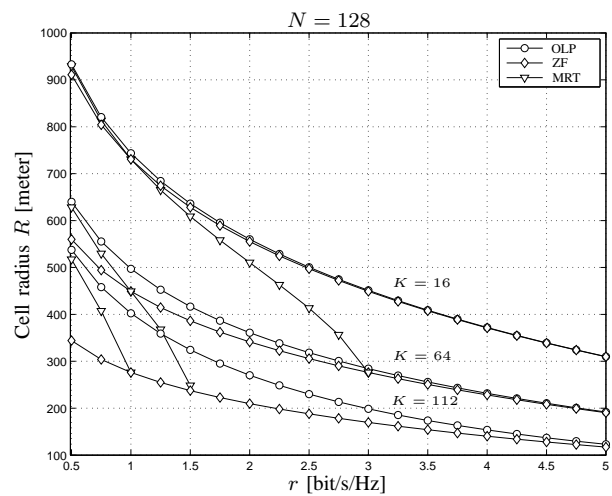}
\caption{\footnotesize{Cell radius obtained with OLP, ZF and MRT to minimize the energy consumption per unit area when $K=16,64$ or $112$.}} \label{fig10}
\end{center} 
\end{figure}

\begin{figure}
\begin{center}
    \psfrag{xlabel}[c][b]{{\footnotesize{Number of users, $K$}}}
    \psfrag{ylabel}[c][t]{\footnotesize{Cell radius $R$ [meter]}}
    \psfrag{data1}[l][m]{\!\!\tiny{OLP}}
    \psfrag{data2}[l][m]{\!\!\tiny{RZF}}
        \psfrag{data3}[l][m]{\!\!\tiny{ZF}}
        \psfrag{title}[c][m]{\!\!\!\!\footnotesize{$N=128$}}
\includegraphics[width=7.8cm]{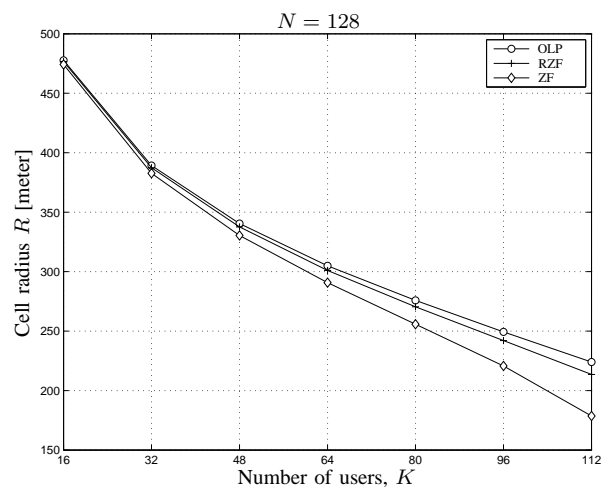}
\caption{\footnotesize{Cell radius obtained with OLP, RZF and ZF to minimize the energy consumption per unit area when $r_k \in [1, 4]$ bit/s/Hz.}} \label{fig11}
\end{center} 
 \end{figure}

Note that the condition $\chi=1\%$ makes the necessary battery level $\mathcal E$ be substantially higher than $\mathsf{E}\left[E_T\right]$.
It is also worth observing that the above value accounts only for the energy required to transmit the signal $\mathbf{s}(t)$ within the time interval $T$. An accurate design of the battery level should also take into account the power amplifier efficiency as well as the energy required for digital signal processing, channel coding and decoding, channel estimation and precoding, and so forth (see \cite{Emil13TW} for more details). However, all these quantities can be somehow quantified off-line and easily added to $\mathcal E$ for a correct design.

\subsection{Optimization of the cell radius}
Assume that a system designer must optimize the cell radius for minimizing the average energy consumption per unit area defined as:
\begin{align}\label{cell_radius}
\mathcal F = \frac{\epsilon + \vartheta T}{R^2}
\end{align}
where $\vartheta$ accounts for the power consumption that is needed for running a BS. Taking the derivative of \eqref{cell_radius} with respect to $R$ yields
\begin{align}
R^\star = \bar x\sqrt[\beta]{\left(1 + 2L_{\bar x}\frac{\eta}{c\underline{\gamma}}\frac{\vartheta}{{\sigma^2}}\right) \frac{\beta+2}{\beta-2}}.
\end{align}
Figs. \ref{fig10} and \ref{fig11} show the optimal cell radius in the same setting of Figs. \ref{fig8} and \ref{fig9}, respectively. {Marks indicate simulation results while solid lines are obtained theoretically.} The parameter $\vartheta$ is set to $18$ Watt and corresponds to the fixed power consumption required for control signals, backhaul, and so forth (this value is taken from \cite{EARTH_D23}). From the results of Figs. \ref{fig10} and \ref{fig11}, it follows that OLP and RZF allow a substantial increase of the coverage area compared to ZF especially when $K$ is large and low rate requirements are imposed.

 \subsection{Imperfect CSI}\label{imperfect_CSI}
{So far, the analysis has been carried out under the assumption of perfect knowledge of $l(\mathbf{x}_{k}(t))$ and ${\bf h}_k(t)$ for $k=1,2,\ldots,K$. Observe that $l(\mathbf{x}_{k}(t))$ corresponds to the average channel attenuation, which changes in time (roughly) three orders of magnitudes slower that the fast fading channels ${\bf h}_k(t)$ \cite{Viering_2002}. More specifically, the characteristic time of change of ${\bf h}_k(t)$ is the time it takes for the UE $k$ to move by a wavelength, i.e. $\sim \lambda \xi/\ell$ while that of $l(\mathbf{x}_{k}(t))$ is $\sim R \xi/\ell$. Therefore, $l(\mathbf{x}_{k}(t))$ changes far slower than ${\bf h}_k(t)$ by a factor $\sim R/\lambda$. In practice, this means that $l(\mathbf{x}_{k}(t))$ maintains constant for a sufficiently large number of reception phases to be accurately estimated at the BS (for example, through measurements of the received signal strength indicator). This makes it reasonable to assume perfect knowledge of $l(\mathbf{x}_{k}(t))$. On the other hand, the fast variations of ${\bf h}_k(t)$ might result into large estimation errors or outdated estimates (especially in high mobility environments) and thus should be taken into account to avoid a severe degradation of the system performance. This is why the impact of imperfect knowledge of ${\bf h}_k(t)$ has been extensively investigated in multi-user MIMO systems either in the finite number of antennas and users regime (e.g., \cite{Vucic2009, Caire2010, Gonzalez2013}) or in the asymptotic regime \cite{Wagner12,Muller13,Sanguinetti_2014_j}. While the power minimization problem with imperfect CSI is still much open for OLP in both regimes, it can be handled for both ZF and RZF. In the asymptotic regime, one might for example use the results illustrated in \cite{Sanguinetti_2014_j}.} 

   \begin{figure}[t!]
\begin{center}
    \psfrag{xlabel}[c][b]{{\footnotesize{$\alpha$ [Watt]}}}
    \psfrag{ylabel}[c][t]{\footnotesize{$\text{Pr}\left(\frac{E_T}{T} > \alpha \right)$}}
    \psfrag{X1}[c][m]{\scriptsize{\text{RZF}}}
    \psfrag{X2}[c][m]{\scriptsize{\text{ZF}}}
    \psfrag{data1}[l][m]{\!\!\tiny{T = 12 [h] -- Sim.}}
    \psfrag{data2}[l][m]{\!\!\tiny{T = 12 [h] -- Theory}}
        \psfrag{data3}[l][m]{\!\!\tiny{T = 3 [h] -- Sim.}}
    \psfrag{data4}[l][m]{\!\!\tiny{T = 3 [h] -- Theory}}
        \psfrag{title}[c][m]{\!\!\!\!\footnotesize{$K=32, N=64, \tau^2 =0.05$}}
\includegraphics[width=8.2cm]{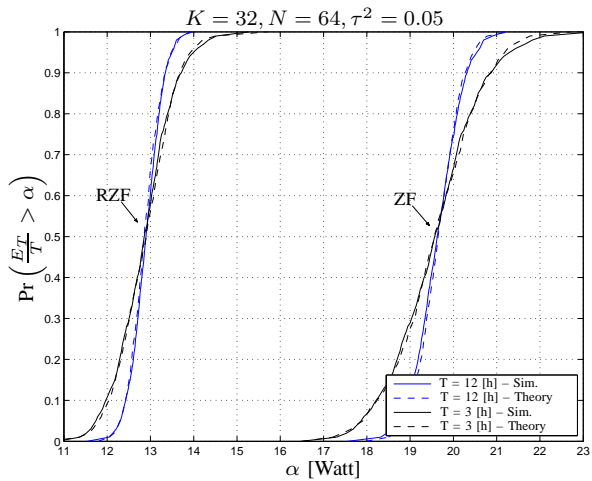}
\caption{\footnotesize{Outage probability $\text{Pr}\left(\frac{E_T}{T}> \alpha \right)$ of RZF and ZF with imperfect CSI, $K=32$, $T=3,12$ hrs and $\tau_i^2 = \tau^2 =0.05$ $\forall i$.}} \label{fig5}
\end{center} 
\end{figure}

{Although applied to a completely different context, the analytical results in \cite{Sanguinetti_2014_j} can be adapted to the system under investigation using standard random matrix theory tools. Omitting the mathematical details, from the results of Section IV in \cite{Sanguinetti_2014_j} the asymptotic power consumption with imperfect CSI takes form:
\begin{align}\label{eq:p}
\overline P(t) = \frac{c\sigma^2}{\eta^\prime } \frac{1}{K} \sum\limits_{i=1}^K \frac{\gamma_i^\prime}{l({\bf x}_i(t))}
\end{align}
where $\gamma_i^\prime=\frac{\gamma_i}{1-\tau_i^2}$ and 
\begin{align}
\eta^\prime = \left\{ {\begin{array}{*{20}{l}}
{1-c-c\sum\limits_{i=1}^K\gamma_i^\prime{\tau_i^2}}&{\text{for ZF}}\\
{1 -c \frac{\sum\limits_{i=1}^K\gamma_i}{K+\sum\limits_{i=1}^K\gamma_i} - c\sum\limits_{i=1}^K\gamma_i^\prime{\tau_i^2}}&{\text{for RZF}}
\end{array}} \right.
\end{align}
with $\tau_i$ accounting for the accuracy or quality of the estimate of ${\bf h}_i(t)$ \cite{Sanguinetti_2014_j}, i.e., $\tau_i=0$ corresponds to perfect CSI, whereas for $\tau_i =1$ the CSI is completely uncorrelated to the true channel. As seen, \eqref{eq:p} has the same general form of \eqref{1010} when imperfect CSI is available. Similar conclusions can be drawn for MRT. To validate these analytical results, Figs.~\ref{fig5} -- \ref{fig6} show the CDF of ${E_T}/{T}$ in the same operating conditions of Fig.~\ref{fig3_A} when $\forall i$ $\tau_i^2= \tau^2 =0.1$ and $\tau_i^2= \tau^2 =0.2$, respectively. The computation of $\epsilon$ and $\Sigma$ is performed through \eqref{avg_pow} and \eqref{var_pow} simply replacing $\eta$ and $\{\gamma_i; i=1,2,\ldots,K\}$ with $\eta^\prime$ and $\{\gamma_i^\prime;i=1,2,\ldots,K\}$. As seen, in both cases the simulation results match with the theoretical ones obtained from Theorem 2 using the asymptotic power given by \eqref{eq:p}. Comparing the results of Figs.~\ref{fig5} -- \ref{fig6} with those of Fig.~\ref{fig3_A}, it is seen that imperfect CSI leads to an increase of the average energy consumption.}

 \begin{figure}
\begin{center}
    \psfrag{xlabel}[c][b]{{\footnotesize{$\alpha$ [Watt]}}}
    \psfrag{ylabel}[c][t]{\footnotesize{$\text{Pr}\left(\frac{E_T}{T} > \alpha \right)$}}
    \psfrag{X1}[c][m]{\scriptsize{\text{RZF}}}
    \psfrag{X2}[c][m]{\scriptsize{\text{ZF}}}
    \psfrag{data1}[l][m]{\!\!\tiny{T = 12 [h] -- Sim.}}
    \psfrag{data2}[l][m]{\!\!\tiny{T = 12 [h] -- Theory}}
        \psfrag{data3}[l][m]{\!\!\tiny{T = 3 [h] -- Sim.}}
    \psfrag{data4}[l][m]{\!\!\tiny{T = 3 [h] -- Theory}}
        \psfrag{title}[c][m]{\!\!\!\!\footnotesize{$K=32, N=64, \tau^2 =0.15$}}
\includegraphics[width=8.2cm]{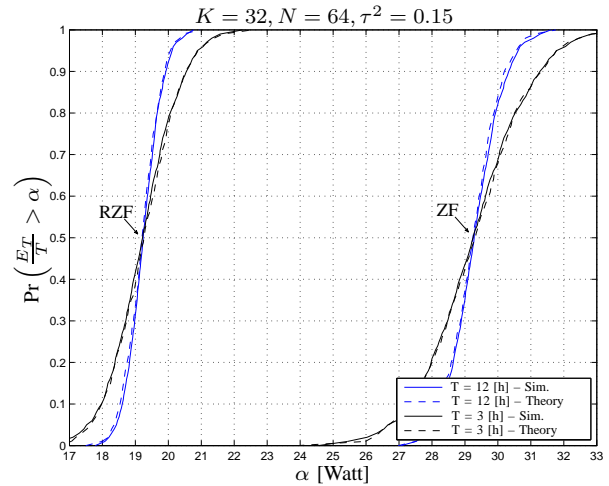}
\caption{\footnotesize{Outage probability $\text{Pr}\left(\frac{E_T}{T}> \alpha \right)$ of RZF and ZF with imperfect CSI, $K=32$, $T=3,12$ hrs and $\tau_i^2 = \tau^2 =0.15$ $\forall i$.}} \label{fig6}
\end{center}
 \end{figure}
 
\section{Conclusions and Discussions}\label{Conclusions}
In this work, we have studied the energy consumption dynamics in the downlink of a single-cell MIMO  network with $N$ antennas at the BS in which $K$ UEs move around according to a random walk model and linear precoding is used to guarantee target rates. The analysis has been conducted for a finite time interval of length $T$ when $K,N\to \infty $ with a fixed ratio when perfect CSI is available at the BS. Under these assumptions, we have shown that the energy consumption converges in distribution to a Gaussian random variable for most of the common precoding techniques and we have computed its mean and variance analytically. We have shown that user mobility plays the dominant role in determining the energy fluctuations. In particular, it turns out that the variance of energy consumption scales as $1/K$ rather than as $1/K^2$ as it happens for motionless UEs. Numerical results have been used to show that the analytical expressions yield accurate approximations when $K$ is of the order of tens (see for example Table \ref{table_values}). As an application of these results, we have dimensioned a battery-powered BS to satisfy a certain probability of running out of energy and we have also computed the cell radius to minimize the energy per unit area. The imperfect CSI has been also briefly addressed.

The analysis has also pointed out that in addition to $K$ there is another parameter that plays a key role in describing the energy fluctuations. This parameter is given by $DT/R^2$, where $R$ is the cell radius and $D$ is the so-called diffusion constant of random walks (or Brownian motions). The quantity $DT/R^2$ roughly describes to what extent the UEs have been around in the cell area. When it is relatively small (as it has been assumed in this work), each UE has not had the chance to visit most of the cell and thus the energy consumed to serve all of them is a random quantity whose fluctuations are basically controlled by $K$. However, when $DT/R^2$ is relatively large, then each UE has likely had the time to move around all the coverage area, and thus the energy consumed to serve each of them becomes approximately deterministic. This result suggests an alternative regime for the asymptotic analysis of energy consumption: the finite $K$ and $N$ but large $T$ (or more properly large $DT/R^2$) regime. This case has been analyzed in some detail in \cite{Moustakas2014_ITW} but it is still pretty much open. { In addition to this, another interesting direction for future work is to consider other mobility models such as the L\'evy flight, which seems to contain some important statistical similarities with human mobility \cite{Rhee2011_LevyHumanMobility, Scafetta2011_LevyWalkHumanMobility}.  The extension of the analysis to a multi-cell network with imperfect CSI and pilot contamination is also an interesting topic for future research.}

\section*{Appendix A\\ Proof of Theorem \ref{lemma:asymptotic-beamforming}}
We start assuming that the SINR constraints are such that the following assumption is satisfied:
\begin{equation}\label{26.11}
\lim_{K \rightarrow \infty} \sup \frac{1}{N}\sum\limits_{i=1}^K \frac{\gamma_i}{1+\gamma_i} < 1.
\end{equation}
Under these circumstances, $\eta$ in \eqref{26} is a positive quantity. Observe that $c\in (0,1)$ is sufficient for the above condition to be satisfied. To proceed further, we rewrite \eqref{23} as follows (using simple calculus)
\begin{align}\label{B1023}
\frac{\gamma_k}{\lambda_k^\star} = \frac{1}{N}{\mathbf{h}_k^H\left(\frac{1}{N}\sum\limits_{i=1,i\ne k}^K\lambda_i^\star\mathbf{  h}_i\mathbf{  h}_i^H + N\mathbf{I}_N\right)^{-1}\mathbf{h}_k}
\end{align}
from which letting
\begin{align}\label{B1}
d_k=\frac{\gamma_k}{\lambda_k^\star l(\mathbf{x}_{k})}
\end{align}
and using \eqref{h_k} we eventually obtain
\begin{align}\label{B2.111}
d_k=\frac{1}{N} {\mathbf{w}_k^H\left(\frac{1}{N}\sum\limits_{i=1,i\ne k}^K\frac{\gamma_i}{d_i}\mathbf{  w}_i\mathbf{  w}_i^H + \mathbf{I}_N\right)^{-1}\mathbf{w}_k}.
\end{align}
Assume that the quantities $\{d_k\}$ are well defined, positive and such that $0\le d_1\le d_2\le\ldots\le d_K$. Then, using monotonicity arguments, from \eqref{B2.111} it follows that
\begin{align}\nonumber
\hspace{-0.3cm}d_K\le\frac{1}{N} {\mathbf{w}_K^H\left(\frac{1}{N}\sum\limits_{i=1}^{K-1}\frac{\gamma_i}{d_K}\mathbf{  w}_i\mathbf{  w}_i^H + \mathbf{I}_N\right)^{-1}\mathbf{w}_K}=\\ = d_K\frac{1}{N} {\mathbf{w}_K^H\left(\frac{1}{N}\sum\limits_{i=1}^{K-1}{\gamma_i}\mathbf{  w}_i\mathbf{  w}_i^H + d_K\mathbf{I}_N\right)^{-1}\mathbf{w}_K}
\end{align}
or, equivalently,
\begin{align}\label{B2}
1\le \frac{1}{N} {\mathbf{w}_K^H\left(\frac{1}{N}\sum\limits_{i=1}^{K-1}{\gamma_i}\mathbf{  w}_i\mathbf{  w}_i^H + d_K\mathbf{I}_N\right)^{-1}\mathbf{w}_K}.
\end{align}
Assume now that $d_K$ is infinitely often larger than $\eta + \ell$ with $\eta$ given by \eqref{26} and $\ell > 0$ some positive value \cite{Romain2014}. Let us restrict ourselves to such a subsequence. From \eqref{B2}, using monotonicity arguments we obtain
\begin{align}\label{B3}
1\le \frac{1}{N} {\mathbf{w}_K^H\left(\frac{1}{N}\sum\limits_{i=1}^{K-1}{\gamma_i}\mathbf{  w}_i\mathbf{  w}_i^H + (\eta + \ell)\mathbf{I}_N\right)^{-1}\mathbf{w}_K}.
\end{align}
Applying standard results in random matrix theory along with the union bound and Markov inequality one gets \cite{Romain2014}
\begin{equation}\label{B4010}
\frac{1}{N} {\mathbf{w}_K^H\left(\frac{1}{N}\sum\limits_{i=1}^{K-1}{\gamma_i}\mathbf{  w}_i\mathbf{  w}_i^H + (\eta + \ell)\mathbf{I}_N\right)^{-1}\mathbf{w}_K}- e(\ell)\rightarrow 0
\end{equation}
with $e(\ell)$ being the unique positive solution to 
\begin{align}
e(\ell)  = \left(\frac{1}{N}\sum\limits_{i=1}^K\frac{ \gamma_i}{1+\gamma_i e(\ell)} + \eta + \ell \right)^{-1}.
\end{align}
From \eqref{B4010}, recalling \eqref{B3} yields 
\begin{align}\label{B7}
\lim_{K \rightarrow \infty} \inf e(\ell) \ge 1.
\end{align}
Using the fact that $e(0) =1$ (recall that $\eta$ is defined as in \eqref{26}) and that $e(\ell)$ is a decreasing function of $\ell$, it can be proved that for any $\ell>0$ \cite{Romain2014}
\begin{align}\label{B8}
\lim_{K \rightarrow \infty} \sup e(\ell) < 1.
\end{align}
This however goes against \eqref{B7} and creates a contradiction on the initial hypothesis that $d_K>\eta + \ell$ infinitely often. Therefore, we must admit that $d_K\le \eta + \ell$
for all large values of $K$. Reverting all inequalities and using similar arguments yields $d_1\ge \eta - \ell$
for all large values of $K$. Putting all these results together yields $\eta - \ell \le d_1\le d_2\le\cdots\le d_K\le \eta + \ell$ from which we may write $\mathop {\max }\nolimits_{k = 1,2,\ldots,K} \left|d_k - \eta \right| \le \ell$
for all large values of $K$ \cite{Romain2014}. Taking a countable sequence of $\ell$ going to zero, we eventually obtain $\mathop {\max }\nolimits_{k = 1,2,\ldots,K} \left|d_k - \eta \right| \rightarrow 0$ from which using \eqref{B1} and assuming that $\lim_{K \rightarrow \infty} \sup \frac{\gamma_k}{l(\mathbf{x}_{k})} < \infty$ the result in \eqref{27} follows.

\section*{Appendix B \\ Random Walk and Brownian Motion}

The convergence of a random walk towards the Brownian motion model is basically due to the central limit theorem. More precisely, observe that ${\bf x}_k(t) = \sum_{i=1}^{\left\lfloor {t/\xi} \right\rfloor} \Delta{\bf x}_k(i)$ has zero-mean and variance $\mathsf{E}\left[\|{\bf x}_k(t)\|^2\right]=t\ell^2/\xi$. Assume now that $t$ is fixed and $\xi\to0$, in the sense that $t/\xi\to\infty$. Thus, the variance of $ {\bf x}_k(t)$ can be kept finite only if the ratio $\ell^2/\xi$ is kept fixed and finite. In this limit, ${\bf x}_k(t)$ becomes a Brownian motion with diffusion coefficient $D$ given by $D = \ell^2/(4\xi)$. In practical systems, however, neither the step size $\ell$ nor the corresponding time $\xi$ vanishes. Nevertheless, the equivalence still holds true since we are basically interested in the long term (corresponding to large values of $t$) and large distance statistics of a random walk with finite $\ell$ and $\xi$ wherein each walker changes position according to a transition rule that only depends on its current position. In these circumstances, the random walk boils down to a Markov process, whose statistical properties in the large time and distance regime are still well captured by simply treating it as a Brownian motion.


\subsection{Transition probability}
The transition probability $\prob({\bf{x}},{\bf{x}}';t-t')$ of a Brownian motion in a given area $\mathcal C$ is defined as the probability of being at a location ${\bf{x}}'$ at time $t$ given that the walker was at ${\bf{x}}$ at time $t'<t$. Mathematically, $\prob({\bf{x}},{\bf{x}}';t-t')$ is obtained solving the following diffusion equation \cite{ItzyksonDrouffe_book}
\begin{equation}\label{eq:diffusion_eq_def}
  \frac{\partial \prob({\bf{x}},{\bf{x}}';t-t')}{\partial t} = {\frac{D}{2}}\nabla^2 \prob({\bf{x}},{\bf{x}}';t-t')
\end{equation}
under appropriate boundary conditions (specifying the behaviour of the user when reaching the boundary $\partial {\mathcal C}$ of ${\mathcal C}$) and subject to:
\begin{align}\label{eq:prob_init_cond_1}
 \int_{\mathcal C} \prob({\bf{x}},{\bf{x}}';t-t') d{\bf{x}} &= 1\\\label{eq:prob_init_cond}
\mathop {\lim }\limits_{t\to t'}\prob({\bf{x}},{\bf{x}}';t-t') & = \delta^2(\bx-\bx').
\end{align}
Note also that $\prob({\bf x},{\bf x}';t-t')=\prob({\bf x}',{\bf x};t-t')$ since the Laplacian operator in \eqref{eq:diffusion_eq_def} is both real and Hermitian.
Although a large variety of boundary conditions can be imposed \cite{morters_book}, the most typical ones are the following two: \emph{i)} the walker exits ${\mathcal C}$ when it reaches the boundary $\partial {\mathcal C}$; \emph{ii)} the walker bounces back into ${\mathcal C}$ when it hits $\partial {\mathcal C}$. The former choice requires to impose $\prob({\bf{x}}\in\partial {\mathcal C}, {\bf{x}}';t-t')=0$ while the latter needs $\br^T\nabla \prob({\bf{x}}\in\partial {\mathcal C},{\bf{x}}';t-t')=0$ with $\br$ being the unit vector along the radial direction.\footnote{Observe that if the walker bounces back, then $\prob({\bf{x}},{\bf{x}}';t-t')$ is always properly normalized in ${\mathcal C}$ (since the probability of exiting ${\mathcal C}$ is zero). This condition is known in the literature as Neuman boundary condition.}

\subsection{Computation of the transition probability}

The transition probability $\prob({\bf{x}},{\bf{x}}'; t-t')$ solving \eqref{eq:diffusion_eq_def} over a bounded domain $\mathcal C$ can be expressed in the following general form \cite{Jackson_EM_book}:
\begin{equation}\label{eq:soln_diff_eq}
  \prob({\bf{x}},{\bf{x}}'; t-t')= \sum_{n}\sum_{m} g^*_{n,m}({\bf{x}})g_{n,m}({\bf{x}}') e^{-\epsilon_{n,m}(t-t')}
\end{equation}
where $g_{n,m}({\bf{x}})$ is the $(n,m)$th (normalized) eigenfunction of the Laplacian operator $\nabla^2$ in ${\mathcal C}$ satisfying the boundary conditions and $\epsilon_{n,m}$ is the corresponding eigenvalue. The above result basically follows from the fact that $\{g_{n,m}({\bf{x}})\}$ form a basis in the Hilbert space $L^2(\mathcal C)$ and thus any function can be written in terms of $\{g_{n,m}({\bf{x}})\}$ (we refer the interested reader to \cite{Jackson_EM_book} for a more detailed discussion on this). Observe that $  \prob({\bf{x}},{\bf{x}}'; t-t')$ in the form of \eqref{eq:soln_diff_eq} meets \eqref{eq:prob_init_cond_1} and \eqref{eq:prob_init_cond} since the eigenfunctions are such that \cite{Jackson_EM_book}:
\begin{align}\label{eq:orthogonality}
  \int_{{\mathcal C}} g^*_{n_1,m_1}({\bf{x}})g_{n_2,m_2}({\bf{x}}) d{\bf x} &= \delta_{n_1,n_2}\delta_{m_1,m_2}\\
  \sum_{n}\sum_{m} g^*_{n,m}({\bf{x}})g_{n,m}({\bf{x}}')  &= \delta^2(\bx-\bx')
\end{align}
where $\delta^2(\bx)$ is the two-dimensional Dirac delta function. The conditions above can also be used to show that if $t''<t'<t$ then
\begin{equation}
\label{eq:Markov}
\int_{{\mathcal C}} \prob({\bf x},{\bf x}';t-t')  \prob({\bf x}',{\bf x}'';t'-t'') d{\bf x}' = \prob({\bf x},{\bf x}'';t-t'')
\end{equation}
which is a consequence of the Markov property of Brownian motions.

Computing the eigenfunctions in closed-form is not always possible for any region ${\mathcal C}$. However, if ${\cal C}$ has a circular symmetry with $|\mathcal C| = \pi R^2$ and the reflecting boundary condition is imposed, then the eigenfunctions can be expressed in polar coordinates $(r,\phi)$ as \cite{Jackson_EM_book}:
\begin{align}\label{eq:eigenfunction_circle}
  g_{n,m}(r,\phi) &= (A_{n,m})^{1/2} J_m\left(\frac{\kappa_{n,m}r}{R}\right) e^{\mathrm{i}m\phi}
\end{align}
where $\kappa_{n,m}$ is the $n$th non-trivial zero of the first derivative of the $m-$Bessel function $J_m(\cdot)$ and $  A_{n,m}^{-1} = \pi R^2 \left(J_m^2(\kappa_{n,m}) + J_{m-1}^2(\kappa_{n,m})\right)$ and the eigenvalue is $\epsilon_{n,m}=\kappa_{n,m}^2$.

\section*{Appendix C \\ Proofs of Lemmas 3 and 4}
In this appendix, we rely on the Brownian motion model (see Appendix B) to compute the mean value and covariance of ${l^{-1}(\mathbf{x}_k{(t)})}$. The former requires to evaluate
\begin{align}\label{E0}
\mathsf{E}_{\mathbf{X}_k}\left[{l^{-1}(\mathbf{x}_k{(t)})}\right] = \mathsf{E}_{\mathbf{x}_k{(0)}}\left[\mathsf{E}_{\mathbf{X}_k|{\mathbf{x}_k{(0)}}}\left[{l^{-1}(\mathbf{x}_k{(t)})}|{\mathbf{x}_k{(0)}}\right]\right]
\end{align}
where $\mathsf{E}_{\mathbf{x}_k{(0)}}\left[z|{\mathbf{x}_k{(0)}}\right]$ denotes the conditional expectation of $z$ with respect to the initial position $\mathbf{x}_k{(0)} \in \mathcal C$. On the other hand, the covariance of ${l^{-1}(\mathbf{x}_k{(t)})}$ takes the form in \eqref{E2.100}.
\begin{figure*}
\begin{align}\nonumber
{\rm{COV}}_{\mathbf{X}_k}\left[{l^{-1}(\mathbf{x}_k{(t)})},{l^{-1}(\mathbf{x}_k{(t'')})}\right] &= \mathsf{E}_{\mathbf{x}_k(0)}\left[\mathsf{E}_{\mathbf{X}_k|\mathbf{x}_k(0)}\left[{l^{-1}(\mathbf{x}_k(t))}{l^{-1}(\mathbf{x}_k(t''))}|\mathbf{x}_k(0)\right]\right] - \\ & - \mathsf{E}_{\mathbf{x}_k(0)}\left[\mathsf{E}_{\mathbf{X}_k|\mathbf{x}_k(0)}\left[{l^{-1}(\mathbf{x}_k(t))}|{\mathbf{x}_k(0)}\right]\mathsf{E}_{\mathbf{X}_k|\mathbf{x}_k(0)}\left[{l^{-1}(\mathbf{x}_k(t''))}|{\mathbf{x}_k(0)}\right]\right].\label{E2.100}
\end{align}
\hrulefill
\vskip-3mm
\end{figure*}
To simplify the notation, in the following derivations we drop the UE index $k$ and relabel $\mathbf{x}_k(t)$, $\mathbf{x}_k(t'')$ and $\mathbf{x}_k(0)$ as follows $\mathbf{x}_k(t) \to \mathbf{x}$, $\mathbf{x}_k(t'') \to \mathbf{x}''$ and $\mathbf{x}_k(0) \to \mathbf{x}'$. In addition, we call $F(\mathbf{x},t|\mathbf{x}',0)$ the probability that the UE reaches $\mathbf{x}$ at a generic time $t$ if its position at time $0$ is $\mathbf{x}'$. From Appendix B, we have that $F(\mathbf{x},t|\mathbf{x}',0)=\prob(\mathbf{x}, \mathbf{x}'; t)$ with $\prob(\mathbf{x}, \mathbf{x}'; t)$ being the solution of the diffusion equation in \eqref{eq:diffusion_eq_def}.

We start computing the mean of ${l^{-1}(\mathbf{x}{(t)})}$. Under the above assumptions, we may write
\begin{align}\label{E0.1}
\mathsf{E}_{\mathbf{X}}\left[{l^{-1}(\mathbf{x}{(t)})}\right]  = \frac{1}{|\mathcal C|}\iint_{\mathcal C} {l^{-1}(\mathbf{x})}\prob(\mathbf{x}, \mathbf{x}'; t)d{\bf x}d{\bf x}'
\end{align}
where we have taken into account that ${\bf x}'$ is uniformly distributed within the circular cell of area $|\mathcal C| = \pi R^2$. Recalling\footnote{This comes from the fact that the probability of reaching (within the same time interval) the point $\mathbf{x}$ starting from $\mathbf{x}'$ is the same as that of reaching $\mathbf{x}'$ starting from $\mathbf{x}$ (see Appendix B).} that $\prob(\mathbf{x},\mathbf{x}';t) = \prob(\mathbf{x}',\mathbf{x};t)$ and using $\int_{\mathcal C}\prob(\mathbf{x}',\mathbf{x};t)d{\bf x}' =1$ (see \eqref{eq:prob_init_cond_1} in Appendix B)
one gets
\begin{align}\label{E0.1}
\mathsf{E}_{\mathbf{X}}\left[{l^{-1}(\mathbf{x}{(t)})}\right]= \frac{1}{|\mathcal C|}\int_{\mathcal C} {l^{-1}(\mathbf{x})}d{\bf x}.
\end{align}
The covariance requires to evaluate the two terms in the right-hand-side of \eqref{E2.100}. Let us start with the first one, which is given by
\begin{align}\label{E3.4}
& \mathsf{E}_{\mathbf{x}'}\left[\mathsf{E}_{\mathbf{X}|\mathbf{x}'}\left[{l^{-1}(\mathbf{x})}{l^{-1}(\mathbf{x}'')}|\mathbf{x}'\right]\right] = \\ & \frac{1}{|\mathcal C|}\iiint_{\mathcal C} {l^{-1}({\bf{x}})}{l^{-1}({\bf{x}}'')} F(\mathbf{x},t,\mathbf{x}'',t''|\mathbf{x}',0)d{\bf x}d{\bf x}''d{\bf x}'.
\end{align}
Thanks to the Markov property of Brownian motions, one gets
\begin{align}\nonumber
F(\mathbf{x},t,\mathbf{x}'',t''|\mathbf{x}',0)  &= F(\mathbf{x},t|\mathbf{x}'',t'',\mathbf{x}',0) F(\mathbf{x}'',t''|\mathbf{x}',0) \\ & =F(\mathbf{x},t|\mathbf{x}'',t'') F(\mathbf{x}'',s|\mathbf{x}',0)\label{E4.100}
\end{align}
or, equivalenty, $F(\mathbf{x},t,\mathbf{x}'',t''|\mathbf{x}',0)=
\prob(\mathbf{x},\mathbf{x}'';t-t'') \prob(\mathbf{x}'',\mathbf{x}';t'').$
Plugging this result into \eqref{E3.4} and using $\prob(\mathbf{x}'',\mathbf{x}';t'') = \prob(\mathbf{x}',\mathbf{x}'';t'')$ with $\int_{\mathcal C}\prob(\mathbf{x}',\mathbf{x}'';t'') d{\bf x}'= 1$ one gets
\begin{align}\label{E6}
&\mathsf{E}_{\mathbf{x}'}\left[\mathsf{E}_{\mathbf{X}|\mathbf{x}'}\left[{l^{-1}(\mathbf{x})}{l^{-1}(\mathbf{x}'')}|\mathbf{x}'\right]\right] = \\ & \frac{1}{|\mathcal C|} \iint_{ \mathcal C} {l^{-1}({\bf{x}})}{l^{-1}({\bf{x}}'')}  \prob(\mathbf{x},\mathbf{x}'';t - t'')d{\bf x}d{\bf x}''.
\end{align}
We are now left with the computation of the second term in \eqref{E2.100}, which is explicitly given by
\begin{align}\label{E8}
& \mathsf{E}_{\mathbf{x}'}\left[\mathsf{E}_{\mathbf{X}|\mathbf{x}'}\left[{l^{-1}(\mathbf{x})}|{\mathbf{x}'}\right]\mathsf{E}_{\mathbf{X}|\mathbf{x}'}\left[{l^{-1}(\mathbf{x}'')}|{\mathbf{x}'}\right]\right] = \\ \label{E8.1}& \frac{1}{|\mathcal C|}\iiint_{ \mathcal C} {l^{-1}({\bf{x}})}{l^{-1}({\bf{x}}'')} \prob(\mathbf{x},\mathbf{x}';t)\prob(\mathbf{x}'',\mathbf{x}';t'')d{\bf x}d{\bf x}''d{\bf x}'.
\end{align}
Using $\prob(\mathbf{x}'',\mathbf{x}';t'') = \prob(\mathbf{x}',\mathbf{x}'';t'')$ and \eqref{eq:Markov} lead to 
\begin{align}\label{E9}
& \mathsf{E}_{\mathbf{x}'}\left[\mathsf{E}_{\mathbf{X}|\mathbf{x}'}\left[{l^{-1}(\mathbf{x})}|{\mathbf{x}'}\right]\mathsf{E}_{\mathbf{X}|\mathbf{x}'}\left[{l^{-1}(\mathbf{x}'')}|{\mathbf{x}'}\right]\right] = \\ &\frac{1}{|\mathcal C|} \iint_{ \mathcal C} {l^{-1}({\bf{x}})}{l^{-1}({\bf{x}}'')} \prob(\mathbf{x},\mathbf{x}'';t + t'')d{\bf x}d{\bf x}''
\end{align}
Plugging the above results together into \eqref{E2.100}, we eventually obtain
\begin{align}\label{E10}
&{\rm{COV}}_{\mathbf{X}}\left[{l^{-1}(\mathbf{x})},{l^{-1}(\mathbf{x}'')}\right]  = \\ &\frac{1}{|\mathcal C|} \iint_{\mathcal C} {l^{-1}({\bf{x}})}{l^{-1}({\bf{x}}'')}Q(\mathbf{x},\mathbf{x}{''};t, t'')d{\bf x}{}d{\bf x}{''}.
\end{align}
where $Q(\mathbf{x},\mathbf{x}'';t, t'')$ is defined as
\begin{align}\label{E11}
& Q(\mathbf{x},\mathbf{x}'';t, t'') = \prob(\mathbf{x},\mathbf{x}'';t - t'') - \prob(\mathbf{x},\mathbf{x}'';t+t'').
\end{align}
We are now left with substituting into \eqref{E10} the closed-form expressions of $\prob(\mathbf{x},\mathbf{x}'';t+t'')$ and $\prob(\mathbf{x},\mathbf{x}'';t - t'')$ as obtained through \eqref{eq:soln_diff_eq} using the eigenfunctions given by \eqref{eq:eigenfunction_circle}. After standard but lengthy computations (not shown for space limitations), we eventually get the result in \eqref{101} of Lemma \ref{lemma_covariance}.

\section*{Appendix D \\ Proof of Theorem \ref{CLT}}

In this appendix, we outline the proof for the central limit theorem of the energy consumption $E_T$. The first step is to observe that since the underlying Brownian motions of UEs are continuous, we may think of the integral in \eqref{E_T} as a limit of a finite sum, i.e.,
\begin{align}\label{eq:E_T_as a sum}
  E_T & = \int_0^T {P(t)dt} =\lim_{L\rightarrow\infty} \frac{T}{L} \sum_{n=1}^L P\left({\bf x}(n),{\bf W}(n)\right)
\end{align}
where ${\bf x}(n)$ and ${\bf W}(n)$ corresponds to ${\bf x}(t)$ and ${\bf W}(t)$ evaluated at $t =nT/L=n\xi$, respectively. In writing the above equation, we have explicitly specified the dependence of $P(t)$ on the UE positions ${\bf x}(n) = \{{\bf x}_k(n); k=1,2,\ldots,K\}$ and on the fading coefficients ${\bf W}(n)=\{{\bf W}_k(n);k=1,2,\ldots,K\}$. We assume that for any $k$ the channel fading coefficients in ${\bf W}_k(n)$ are independent over $n$ since the fast fading channel ${\bf W}_k(t)$ decorrelates approximately as $\sim \lambda\xi/\ell$, which is much smaller than the sampling time $\xi$.

%
%
%
%
To proceed further, we let $\Delta E_T = E_T  - \mathsf{E}\left[E_T \right]$ be the difference between $E_T$ and its average value $\mathsf{E}\left[E_T \right]$ and rewrite $\Delta E_T$ as the sum of two terms, i.e., $\Delta E_T = \Delta_1 + \Delta_2$ with
\begin{align}\label{eq:diff_ET}
\Delta_1 &= E_T - \int_0^T \overline{P}(t) dt \\   \Delta_2 &= \int_0^T  \overline{P}(t) dt -  \mathsf{E}\left[E_T \right]
\end{align}
Observe that the variance of $\Delta_1$ is just $T$ times the variance of $P(t)-{\overline P}(t)$ averaged over user locations, i.e.,
\begin{align}
\mathsf{VAR}\left[\Delta_1\right]=T\mathsf{E}_X\left[\mathsf{VAR}_{\bf W}\left[P(t)-{\overline P}(t)\right]\right].
\end{align}
In \cite{Bai2004_CLT_covariance_matrices}, the authors show that $K\mathsf{VAR}\left[P(t)-{\overline P}(t)\right]$ is finite in the limit $K\to\infty$. Since the pathloss function is bounded from below, it follows that $K\mathsf{E}_X\left[\mathsf{VAR}\left[P(t)-{\overline P}(t)\right]\right]$ is also bounded, i.e.,
\begin{equation}\label{eq:supvar}
  \lim_{K\to \infty} \sup K\mathsf{E}_X\left[\mathsf{VAR}\left[P(t)-{\overline P}(t)\right]\right] <\infty.
\end{equation}
This means that the fluctuations of $\Delta_1$ are of order $O(K^{-1})$.

To analyze the second term $\Delta_2$, we use \eqref{1010} and observe that the integral in \eqref{eq:diff_ET} can be written as a sum over independent Brownian paths of the $K$ users in the system. Using the results of Appendix C, the term corresponding to $k$th UE has mean and variance given by
\begin{align}\label{eq:app:meanE_k}
{\frac{1}{K}\frac{c\sigma^2}{\eta} \int_0^T \mathsf{E}_{\mathbf{X}_k}\left[\frac{\gamma_k}{\ell({\bf x}_k(t))}\right]  = \frac{\gamma_k}{|\mathcal C|}\frac{T}{K}\frac{c\sigma^2}{\eta}\int_{\mathcal C} \frac{1}{l(\mathbf{x}_k{(\tau)})}d{\bf x}_k{(\tau)}}
\end{align}
and
\begin{align}\nonumber
\frac{1}{K^2} \left(\frac{c\sigma^2}{\eta}\right)^2\iint_0^T &{\mathsf{COV}}_{\mathbf{X}_k}\left[\frac{\gamma_k}{\ell({\bf x}_k(\tau))},\frac{\gamma_k}{\ell({\bf x}_k(s))}\right] = \\ & =\gamma_k^2 \frac{TR^2}{DK^2} \left(\frac{c\sigma^2}{\eta}\right)^2\Theta
\label{eq:app:varE_k}
\end{align}
where $\Theta$ is obtained as in \eqref{105}. Hence, $\Delta_2$ has finite mean and fluctuations of order $K^{-1/2}$ and its variance can be obtained summing \eqref{eq:app:varE_k} over $k$. This leads to 
\begin{align}\label{var_pow_Appendix}
\mathsf{VAR}\left[\Delta_2\right] &= \left(\frac{c\sigma^2}{\eta}\right)^2\left(\frac{1}{K}\sum\limits_{i=1}^K\gamma_i^2\right) \frac{T R^{2}}{KD} \Theta.
\end{align}
As a result of the above, we eventually have that
\begin{align}\label{eq:CLT}
\sqrt{K}\left(\frac{E_T - \epsilon }{\sqrt{\Sigma}} \right)\mathop {\longrightarrow}  \limits_{K,N \to \infty}^{\mathcal D} \mathcal{N}(0,1)
\end{align}
as stated in Theorem \ref{CLT}. As seen, the proof is basically articulated in two steps. First, we exploit the results in \cite{Bai2004_CLT_covariance_matrices} to point out that that the power variance due to the fast fading scales as $O(K^{-2})$ for all considered schemes. Then, we take advantage of \eqref{1010} and use the results in Lemmas \ref{lemma_mean} and \ref{lemma_covariance} to deal with the UE movements and to prove the results. 
%
%
\label{Appendix_D}

\section*{Acknowledgment}
The authors thank Dr.~Romain Couillet for helpful discussions on the large system analysis of the optimal linear precoding and in particular for the proof of Theorem 1.

\ifCLASSOPTIONcaptionsoff
  \newpage
\fi

\bibliographystyle{IEEEtran}
\bibliography{IEEEabrv,refs}

\end{document}